\documentclass{article}
\usepackage[a4paper, total={6in, 8in}]{geometry}

\usepackage{amsmath}
\usepackage{amsthm}
\usepackage{amssymb}
\usepackage[toc,page]{appendix}
\usepackage{bm}
\usepackage{epstopdf}
\usepackage{subfigure}
\usepackage{threeparttable}
\usepackage{graphicx}
\usepackage{booktabs}
\usepackage{pdflscape}
\usepackage{multirow}
\usepackage{float}
\usepackage{xcolor}
\usepackage{natbib}
\usepackage{tabularx}
\usepackage{bbm}
\usepackage{easyReview}
\usepackage{caption}

\makeatletter
\renewcommand*\env@matrix[1][c]{\hskip -\arraycolsep
	\let\@ifnextchar\new@ifnextchar
	\array{*\c@MaxMatrixCols #1}}
\makeatother

\newcolumntype{Y}{>{\centering\arraybackslash}X}

\newtheorem{theorem}{Theorem}
\newtheorem{assumption}{Assumption}

\newtheorem*{assumption*}{\assumptionnumber}

\providecommand{\assumptionnumber}{}
\makeatletter

\usepackage{authblk}

\title{A sequential test procedure for the choice of the number of regimes in multivariate nonlinear models}

\author[1]{Andrea Bucci}
\affil[1]{Department of Economics and Law, University of Macerata, Italy}

\date{}
\begin{document}
\maketitle

\begin{abstract}
This paper proposes a sequential test procedure for determining the number of regimes in nonlinear multivariate autoregressive models. The procedure relies on linearity and no additional nonlinearity tests for both multivariate smooth transition and threshold autoregressive models. We conduct a simulation study to evaluate the finite-sample properties of the proposed test in small samples. Our findings indicate that the test exhibits satisfactory size properties, with the rescaled version of the Lagrange Multiplier test statistics demonstrating the best performance in most simulation settings. The sequential procedure is also applied to two empirical cases, the US monthly interest rates and Icelandic river flows. In both cases, the detected number of regimes aligns well with the existing literature.

\end{abstract}


\normalsize

\newpage
\section{Introduction}

Linear vector autoregressive models (VAR) have been a cornerstone in the analysis of multivariate time series for over four decades since the seminal paper by \cite{Sims1980}. Despite their widespread use, linear models often fail to capture the complexity of real-world data, particularly when relationships exhibit nonlinear dynamics. For example, financial asset prices respond asymmetrically to unexpected macroeconomic news \citep{Anderson1999}, hence requiring models that can accommodate such nonlinear behaviors.

Advances in computational power have facilitated the development of more complex models, such as the vector logistic smooth transition regression (VLSTR) and the vector threshold regression (VTR). These models offer greater flexibility by allowing for regime changes based on the value of a transition variable, see \cite{hute13} for a comprehensive review. Despite the increasing popularity of these models, their empirical application on real problems is yet limited, partly due to the challenges in model specification and the lack of robust tests for linearity and misspecification. 

Therefore, proper specification tests can be crucial for these models, which are not identified if the linear or a lower-regime model is the data-generating process \citep{Davies1987}. In the univariate context, a linearity test has been developed by \cite{lusate88}, while \cite{Eitrheim1996} construct misspecification tests for smooth transition autoregressive (STAR) models, including an error autocorrelation test, a test of no additional nonlinearity and a test against parameter non-constancy. For multivariate models, \cite{camacho04} and \cite{yate14} have extended these tests. The former has developed a modelling strategy for a bivariate VLSTAR model along with several misspecification tests employing an equation-by-equation approach. The latter builds upon Camacho's approach by extending the linearity and misspecification tests to a system-based approach and generalizing the tests beyond two time series. This provides additional flexibility in modelling complex multivariate nonlinear dependencies and permits capturing the possible nonlinear interactions between the variables, which may be missed in an equation-by-equation approach.

In this context, determining the number of regimes is not straightforward. Several attempts to identify the number of regimes have been made in the univariate framework in \cite{Hansen1999}, \cite{Gonzalo2002}, and \cite{Strikholm2006}. Inspired by the sequential test for structural breaks in \cite{baipe98}, these last two approaches suggest choosing the number of regimes starting from a linear model, \textit{i.e.} a single regime model, and testing iteratively between $m$ and $m+1$ regimes until rejection of the null hypothesis of $m$ regimes. This paper aims to fill the gap in the existing literature, proposing a easy-to-implement sequential procedure for the selection of the number of regimes in multivariate problems. The sequential procedure is a mere extension of the approach proposed in \cite{Strikholm2006} and applies both to smooth and abrupt regime-changing models. As for the univariate version in \cite{Strikholm2006}, the practitioner has full control of the asymptotic significance level of the test at each step.  We demonstrate that the finite-sample properties of the test procedure are satisfying either if the data are generated from a VLSTAR or a Threshold Vector Autoregressive model (TVAR).

One of the possible challenges of a system-based approach for the tests is the existence of stationarity and ergodicity conditions for the VLSTAR model. Although the papers from \cite{Saikkonen2008} and \cite{Kheifets2020} provide the conditions for stationarity and ergodicity in particular cases, the conditions for the general model are not available, therefore the test procedure proposed here works asymptotically properly only for a single-transition null hypothesis.

To validate our approach, we apply the sequential procedure to two empirical problems. On the one hand, we try to detect the number of regimes in US monthly interest rates \citep{Tsay1998}. On the other hand, the sequential procedure is applied to daily Icelandic river flow data, which have shown to be nonlinear in several former applications \citep{Tong1985, Tsay1998, teya14, LivingstonJr2020}. In both cases, the number of regimes detected overlaps with what was found in the related literature.

The paper is organized as follows. Section \ref{sec:VLSTAR} describes the vector logistic smooth transition autoregressive model. In Section \ref{sec:Linearity}, we define the linearity test, while the sequential test procedure is introduced in Section \ref{sec:Nregimes}. The tests are then applied to simulated data in Section \ref{sec:Simulation} to compute their empirical sizes, empirical powers and selection frequencies, and to real data in Section \ref{sec:Empirical}. Section \ref{sec:Conclusions} concludes.

\section{The VLSTAR model}\label{sec:VLSTAR}
A specification for the general VLSTR model can be found in \cite{teya14}. For ease of notation, in this study we do not include exogenous variables in the model, this means that we are analysing a vector logistic smooth transition autoregressive (VLSTAR) model. Let $\mathbf{y}_t$ be an $n \times 1$ vector of dependent variables, the VLSTAR model with $m$ regimes can be defined as follows:
\begin{align}\label{eq:VLSTAR}
	\mathbf{y}_t &= \bm{\mu}_0 + \sum_{j = 1}^{p}\mathbf{\Phi}_{0,j}\mathbf{y}_{t-j} + \mathbf{G}_t^{(1)}\left(\mathbf{s}_t; \bm{\gamma}_{1}, \mathbf{c}_1\right)\left[\bm{\mu}_1 + \sum_{j = 1}^{p}\mathbf{\Phi}_{1,j} \mathbf{y}_{t-j}\right] +\nonumber \\
	&\dots + \mathbf{G}_t^{(m-1)}\left(\mathbf{s}_t; \bm{\gamma}_{m-1}, \mathbf{c}_{m-1}\right)\left[\bm{\mu}_{m-1} + \sum_{j = 1}^{p}\mathbf{\Phi}_{m-1,j} \mathbf{y}_{t-j}\right] + \bm{\varepsilon}_t
\end{align}
where $\bm{\mu}_d$ is an $n \times 1$ vector of intercepts, for $d = 0, \ldots, m-1$, $\mathbf{\Phi}_{d,j}$ is an $n \times n$ matrix of parameters for the $j$-th lag and $\mathbf{G}_t^{(d)}\left(\mathbf{s}_t; \bm{\gamma}_{d}, \mathbf{c}_{d}\right)$ is a diagonal matrix of transition functions such that
\begin{equation}\label{eq:Gt}
	\mathbf{G}_t^{(d)}\left(\mathbf{s}_t; \bm{\gamma}_{d}, \mathbf{c}_{d}\right) = \text{diag}\left\{ g_{1,t}^{(d)}\left(s_{1,t}; \gamma_{1d}, c_{1d}\right), \ldots, g_{n,t}^{(d)}\left(s_{n,t}; \gamma_{nd}, c_{nd}\right)\right\}
\end{equation}
where $s_{i,t}$, for $i = 1, \ldots, n$, is a weakly stationary transition variable for the $i$-th equation, while $\gamma_{i d}$ and $c_{id}$ are respectively the slope parameter and the location parameter where the transitions occur for the $d$-th regime. The lagged values of $\mathbf{y}_t$, or a combination of them \citep{camacho04, Kheifets2020}, are usually chosen as $s_{i,t}$ for smooth transition models. However, a stationary exogenous variable can also be used, and, according to \cite{He2008}, a temporal trend such as $s_{i,t} = t/T$ can be employed as well without violating the asymptotic theory. In this case, the VLSTAR model can be considered a special case of a time-varying autoregressive (TV-VAR) model and, for $\gamma_{i,d} \rightarrow \infty$, the changes of regimes identify structural breaks in the model. This could provide a good alternative to already existing methods for the identification of co-shifting in multivariate time series \citep{Hendry1998}.

The elements of $\mathbf{G}_t^{(d)}$ in Eq. \eqref{eq:Gt} are usually specified as standard logistic functions\footnote{We use standard logistic functions because of their simplicity, but a more general version of the logistic function can also be used \citep{He2008}.}
\begin{equation*}
	g_{i,t}^{(d)}\left(s_{i,t}; \gamma_{id}, c_{id}\right) = \left[1 + \exp\left\{-\gamma_{id}\left(s_{i,t}-c_{id}\right)\right\}\right]^{-1}, \quad \gamma_{id} >0.
\end{equation*}
This specification is extremely flexible, since for $\gamma_{d} \rightarrow \infty$, $\forall d$, the diagonal elements of $\mathbf{G}^{(d)}_t\left(\mathbf{s}_t; \bm{\gamma}_{d}, \mathbf{c}_{d}\right)$ (for ease of notation we will refer to this function as $\mathbf{G}^{(d)}_t$) approach the indicator function, $\mathbbm{1}(s_{i,t}> c_{id})$, thus the model becomes a vector threshold autoregressive (VTAR) model as the one introduced by \cite{Tsay1998}, while for $\gamma_{d} \rightarrow 0$, the model becomes a simple VAR. This means that the approach proposed in this study based on a smooth transition model can also be implemented for the selection of the number of regimes in a VTAR, for $\gamma_{d}$ sufficiently large (see Section \ref{sec:VTAR} for a discussion).

Model \eqref{eq:VLSTAR} can be reparametrized in the following form
\begin{equation}\label{eq:VLSTAR2}
	\mathbf{y}_t = \left\{\sum_{d=1}^{m}\mathbf{G}_t^{(d-1)}\mathbf{B}_r'\right\}\mathbf{x}_t + \bm{\varepsilon}_t = \left[\mathbf{I}_n \ \mathbf{G}_t^{(1)} \ldots \ \mathbf{G}_t^{(m-1)}\right] \begin{bmatrix}
		\mathbf{B}_1 \\ \mathbf{B}_2 \\ \vdots \\ \mathbf{B}_m
	\end{bmatrix} \mathbf{x}_t + \bm{\varepsilon}_t = \mathbf{\Psi}_t' \mathbf{B}' \mathbf{x}_t + \bm{\varepsilon}_t
\end{equation}
where $\mathbf{\Psi}_t = \left(\mathbf{I}_n, \mathbf{G}_{t}^{(1)}, \ldots, \mathbf{G}_t^{(m-1)}\right)'$ is a $m n \times n$ matrix, $\mathbf{I}_n$ is an $n \times n$ identity matrix, $\mathbf{x}_t = \left[1, \mathbf{y}_{t-1}', \mathbf{y}_{t-2}', \ldots, \mathbf{y}_{t-p}'\right]'$ is a $(1+p n) \times 1$ vector and $\mathbf{B} = \left(\mathbf{B}_1, \mathbf{B}_2, \ldots, \mathbf{B}_m\right)$ is a $(1 + p  n) \times m  n$ matrix of parameters, where $\mathbf{B}_d = \left(\bm{\mu}_d', \mathbf{\Phi}_{d,1}', \ldots, \mathbf{\Phi}_{d,p}'\right)'$. Setting $\mathbf{G}_t^{(0)} = \mathbf{I}_n$ indicates that no transitions are allowed before the first change of regime. The set of parameters to be estimated is $\bm{\theta} = \left\{\mathbf{B}, \bm{\Gamma}, \mathbf{C}\right\}$, where $\bm{\Gamma}$ and $\mathbf{C}$ are $n \times m$ matrices of parameters of the transition functions. 

The linearity and additive nonlinearity testing problems in the model \eqref{eq:VLSTAR} concern testing the additive ($m-1$)-th component, therefore the null hypothesis is that $\bm{\mu}_{m-1} = \mathbf{0}$, $\mathbf{\Phi}_{m-1,j} = \mathbf{0}$, $j = 1, \ldots, p$, in which case $\mathbf{G}_t^{(m-1)}$ is not identified since it contains unidentified nuisance parameters. Equivalently, the null hypothesis can be specified as $H_0 \colon \gamma_{i, m-1} = 0$, consequently $\mathbf{G}_t^{(m-1)} = (1/2)\mathbf{I}_n$, where $\mathbf{I}_n$ is an $n \times n$ identity matrix. This implies that the model is not identified because the linear component contains too many parameters that cannot be estimated consistently. The fact that a null hypothesis can be specified in different ways indicates a lack of identification of model \eqref{eq:VLSTAR} under the null hypothesis. This problem, firstly studied by \cite{Davies1987} and \cite{Watson1985}, has the direct consequence that the standard asymptotic inference does not hold as the asymptotic distribution of the test is not known under the null. To overcome it, \cite{Hansen1996} has provided an empirical null distribution by simulation and has given the asymptotic theory for inference. Nevertheless, this method is computationally demanding and applies only in the case of a common transition function among all the equations, $\mathbf{G}_t = g(s_t|\gamma,c)\mathbf{I}_n$, so the number of nuisance parameters is restricted to two. Alternatively, the use of a Taylor series approximation around the null and a Lagrange multiplier (LM) test has been used to circumvent the identification problem, see \cite{lusate88}, \cite{tera94} for the univariate smooth transition model. Recently, \cite{Seong2022} consider testing both the null hypotheses in a univariate smooth transition model and combining the results in a single quasi-likelihood ratio test statistic \citep{Cho2011a, White2012}. Following \cite{lusate88} and \cite{Strikholm2006}, we propose to approximate the logistic function in the alternative hypothesis through a $L$-order Taylor approximation around $\gamma_i = 0$, as further discussed in Section \ref{sec:Linearity}.

\section{Linearity test}\label{sec:Linearity}

In our sequential procedure the linearity test is the first step, since the smooth transition model is not identified if the linear model is the true data-generating process. When the system foresees a different transition variable for each equation, linearity can be tested equation-by-equation through the test introduced by \cite{lusate88}. Otherwise, the joint linearity test introduced in \cite{yate14} can be performed when a single transition variable is used. In the next sections, we deepen the theory behind the linearity and no additional nonlinearity tests already proposed in \cite{yate14}.

\subsection{Testing linearity with a common transition variable}
The asymptotic normality of the score used to derive the test statistic is guaranteed under the regularity conditions provided by \cite{Basawa1976} and the Assumptions in the following Section.

By considering a 2-regime model (\textit{i.e.}, $m = 2$), Eq. \eqref{eq:VLSTAR2} becomes
\begin{equation}\label{eq:VLSTAR3}
\mathbf{y}_t = \mathbf{B}_1' \mathbf{x}_t + \mathbf{G}_t \mathbf{B}_2' \mathbf{x}_t + \bm{\varepsilon}_t.
\end{equation}
Testing linearity in Eq. \eqref{eq:VLSTAR3} equals testing the null hypothesis $\text{H}_0: \gamma_i = 0$, $i = 1, \ldots, n$. Under the null, we have that $\mathbf{G}_t = \left(1/2\right)\mathbf{I}_n$ and that Eq. \eqref{eq:VLSTAR3} is linear, meaning that the null hypothesis creates an identification problem for the parameters in the linear combination $\mathbf{B}_1 + (1/2) \mathbf{B}_2$ and for the location parameter, $c_i$. As already pointed out above, this identification problem can be overcome by approximating the logistic function through an $L$-order Taylor approximation around $\gamma_i = 0$, such that
\begin{equation*}
g_{i,t}\left(s_t \mid \gamma_i, c_i\right) \approx \sum_{l = 0}^{L}  \upsilon_{i, l}s_t^l + r_{i,t}
\end{equation*}
where $\upsilon_{i,0}, \ldots, \upsilon_{i,L}$ are the coefficients and $r_{i,t}$ is the reminder term. This means that $\mathbf{G}_t$ can be written as follows:
\begin{align}\label{eq:approxGt}
\mathbf{G}_t &\approx \text{diag}\left\{\sum_{l=0}^{L}\upsilon_{1, l}s_t^{l} + r_{1,t}, \ldots, \sum_{l=0}^{L}\upsilon_{n, l}s_t^{l} + r_{n,t} \right\} \nonumber\\
&\approx \sum_{l = 0}^{L} \mathbf{\Upsilon}_{l} s_t^{l} + \mathbf{R}_t
\end{align}
where $\mathbf{\Upsilon}_{l} = \text{diag}\left(\upsilon_{1,l}, \ldots, \upsilon_{n, l}\right)$ and $\mathbf{R}_t = \text{diag}\left(r_{1,t}, \ldots, r_{n,t}\right)$. Inserting Eq. \eqref{eq:approxGt} in \eqref{eq:VLSTAR3} yields:
\begin{align}\label{eq:approxVLSTAR}
\mathbf{y}_t &= \mathbf{B}_1'\mathbf{x}_t + \left(\sum_{l=0}^{L}\mathbf{\Upsilon}_{l}s_t^{l} + \mathbf{R}_t\right)\mathbf{B}_2'\mathbf{x}_t + \bm{\varepsilon}_t \nonumber\\
&= \left(\mathbf{B}_1' + \mathbf{\Upsilon}_0\mathbf{B}_2'\right)\mathbf{x}_t + \sum_{l = 1}^{L}\mathbf{\Upsilon}_l \mathbf{B}_2' \mathbf{x}_t s_t^l + \mathbf{R}_t \mathbf{B}_2'\mathbf{x}_t + \bm{\varepsilon}_t \nonumber\\
&= \mathbf{D}_0'\mathbf{x}_t + \sum_{l = 1}^{L}\mathbf{D}_l'\mathbf{x}_t s_t^l + \bm{\varepsilon}_t^*
\end{align}
where $\mathbf{D}_0 = \mathbf{B}_1 + \mathbf{B}_2 \mathbf{\Upsilon}_0'$, $\mathbf{D}_l = \mathbf{B}_2 \mathbf{\Upsilon}_l'$ and $\bm{\varepsilon}_t^* = \mathbf{R}_t \mathbf{B}_2'\mathbf{x}_t + \bm{\varepsilon}_t$. In the auxiliary VAR in Eq. \eqref{eq:approxVLSTAR}, testing linearity is equal to testing the null hypothesis $\text{H}_0\colon \mathbf{D}_1 = \dots = \mathbf{D}_L  = \mathbf{0}$. Under the null hypothesis $\mathbf{R}_t = \mathbf{0}$, therefore the error term is $\bm{\varepsilon}_t^* = \bm{\varepsilon}_t$, so that the distributional properties of the error process are not affected by the Taylor approximation under the null hypothesis.

Denoting $\mathbf{Y} = \left(\mathbf{y}_1, \ldots, \mathbf{y}_T\right)'$, $\mathbf{X} = \left(\mathbf{x}_1, \ldots, \mathbf{x}_T\right)'$, $\mathbf{E}^* = \left(\bm{\varepsilon}_1^*, \ldots, \bm{\varepsilon}_T^*\right)'$, $\mathbf{\tilde{D}}_L = \left(\mathbf{D}_1', \ldots, \mathbf{D}_L'\right)'$, and
\begin{equation}\label{eq:Z}
\mathbf{Z}_L = \begin{bmatrix}
\mathbf{x}_1's_1 & \mathbf{x}_1's_1^2 & \ldots &\mathbf{x}_1's_1^L\\
\mathbf{x}_2's_2 & \mathbf{x}_2's_2^2 & \ldots &\mathbf{x}_2's_2^L\\
\vdots & \vdots & \ddots & \vdots\\
\mathbf{x}_T's_T & \mathbf{x}_T's_T^2 & \ldots & \mathbf{x}_T's_T^L
\end{bmatrix},
\end{equation}
Eq. \eqref{eq:approxVLSTAR} can be written as
\begin{equation}\label{eq:VLSTAR4}
\mathbf{Y} = \mathbf{X} \mathbf{D}_0 + \mathbf{Z}_L \mathbf{\tilde{D}}_L + \mathbf{E}^*.
\end{equation}

The null hypothesis is $\mathbf{\tilde{D}}_L = \mathbf{0}$, while the subscript in $\mathbf{Z}$ and $\mathbf{\tilde{D}}$ indicates the order of the Taylor expansion.

Let $\bm{\theta} = \left(\mathbf{d}_0',\mathbf{d}_1'\right)' \in \Theta$ be the unknown parameters of the model \eqref{eq:VLSTAR4} with the true values $\bm{\theta}_0$, where $\mathbf{d}_0 =\text{vec}(\mathbf{D}_0)$, $\mathbf{d}_1 = \text{vec}(\mathbf{\tilde{D}}_L)$, $\Theta = \Theta_{\mathbf{d}_0} \times \Theta_{\mathbf{d}_1}$ is the parametric space with $\Theta_0 \in \mathbbm{R}^{\tau_0}$ and $\Theta_1 \in \mathbbm{R}^{\tau_1}$, with $\tau_0 = (1+pn)n$ and $\tau_1 = (1+pn)n + 2n$. Below, we assume that $\Theta_0$ and $\Theta_1$ are compact and $\bm{\theta}_0$ is an interior point of $\Theta$. To compute a test for the null hypothesis, the log-likelihood of model \eqref{eq:VLSTAR4} for $T$ observations must be specified as follows
\begin{equation}\label{eq:LL}
	\ell_T(\bm{\theta}) = \sum_{t=1}^{T}\ell_t(\bm{\theta}) = k - \frac{1}{2}\sum_{t=1}^{T}\log |\mathbf{\Omega}_t| - \frac{1}{2}\sum_{t=1}^{T}\bm{\varepsilon}_t'\mathbf{\Omega_t}^{-1}\bm{\varepsilon}_t
\end{equation}
where 
\begin{equation*}
\bm{\varepsilon}_t = \mathbf{y}_t - \mathbf{D}_0'\mathbf{x}_t - \sum_{l = 1}^L\mathbf{D}_l'\mathbf{x}_ts_t^l = \mathbf{y}_t - \mathbf{D}_0'\mathbf{x}_t - \mathbf{\tilde{D}}_L'\mathbf{z}_t,
\end{equation*}
with $\mathbf{z}_t = \left(\mathbf{x}_t's_t, \mathbf{x}_t's_t^2, \ldots, \mathbf{x}_t's_t^l\right)'$ and $E\left\{\bm{\varepsilon}_t\bm{\varepsilon}_t' | \mathcal{F}_{t-1}\right\} = \mathbf{\Omega}_t$ is a positive definite covariance matrix, with $\lim_{T \rightarrow \infty}(1/T)\sum_{t=1}^{T}\mathbf{\Omega}_t = \mathbf{\Omega}$, see the following Assumption 3 for further details. Consequently, the limiting covariance matrix can be estimated from $(1/T)\sum_{t=1}^{T}\hat{\bm{\varepsilon}}_t\hat{\bm{\varepsilon}}_t'$ and can be used in the construction of the test statistic.

We need to specify the following assumptions in order to define an LM test.

\begin{assumption}
	The log-likelihood $\ell_T(\bm{\theta})$, defined as in Eq. \eqref{eq:LL}, is twice continuously differentiable with respect to $\bm{\theta}$ in an open neighbourhood of $\mathbf{D}_1 = \mathbf{0}$.
\end{assumption}

\begin{assumption}
	The maximum likelihood estimators of the parameters $\mathbf{D}_0$ are consistent under the null hypothesis $\mathbf{D}_1 = 0$.
\end{assumption}
\begin{assumption}
	The stochastic sequence $\left\{\bm{\varepsilon}_t\right\}$ is a martingale difference sequence with respect to an increasing sequence of  $\sigma$-fields, $\mathcal{F}_t$ with
	\begin{equation*}
		\underset{t}{\sup} \ E\left\{|\varepsilon_{i,t}|^{2+\alpha}| \mathcal{F}_{t-1}\right\} < \infty \qquad \text{a.s.}
	\end{equation*}
	for some $\alpha > 0$ and $i = 1, \ldots, n$, with $E\left\{\bm{\varepsilon}_t\bm{\varepsilon}_t'|\mathcal{F}_{t-1}\right\} = \mathbf{\Omega}_t$, where $\mathbf{\Omega}_t$ is a positive definite matrix with the following asymptotic limit
	\begin{equation*}
		\lim_{T\rightarrow\infty}(1/T)\sum_{t=1}^{T}\mathbf{\Omega}_t = \mathbf{\Omega} \qquad \text{a.s.}
	\end{equation*}
	for some positive definite matrix $\mathbf{\Omega}$.
\end{assumption}

\begin{assumption}
$\mathbf{X'X}$ and $\mathbf{Z}_L'(\mathbf{I} - \mathbf{P}_{\mathbf{X}})\mathbf{Z}_L$, where $\mathbf{P}_X$ is the limiting projection matrix of $\mathbf{X}$, $\mathbf{P}_{\mathbf{X}} = \mathbf{X}(\mathbf{X'X})^{-1}\mathbf{X}'$, are positive definite matrices.
\end{assumption}

Assumption 2 is a high-level assumption, while Assumption 3 guarantees the existence of the second moments for $\mathbf{y}_t$ and the convergence of the sample moments to their true values \citep{He2008} and permits the use of asymptotic theory for a martingale difference sequence (MDS), even when the assumption of i.i.d. errors is not valid, \textit{e.g.}, in the case of conditionally heteroskedastic errors \citep{Wang2022}. Assumption 4 is a moment condition: for instance, if the model is a VLSTAR and $s_t = y_{i,t-d}$, $d >0$, this implies that $\mathbf{y}_t$ has a finite $2(L+1)$-th moment.

The block of the score vector involving the parameters under test, $\tilde{\bm{\theta}}$, can be written as follows

\begin{align}
\frac{\partial \ell_T (\tilde{\bm{\theta}})}{\partial \mathbf{\tilde{D}}_L} &= -\frac{\partial}{\partial \mathbf{\tilde{D}}_L}(1/2)\sum_{t=1}^T \bm{\varepsilon}_t'\bm{\Omega}^{-1}\bm{\varepsilon}_t = \frac{\partial}{\partial \mathbf{\tilde{D}}_L}\sum_{t=1}^T\mathbf{z}_t' \mathbf{\tilde{D}}_L \bm{\Omega}^{-1}\bm{\varepsilon}_t \nonumber\\
&= \sum_{t=1}^{T}\mathbf{z}_t \bm{\varepsilon}_t' \bm{\Omega}^{-1} = \mathbf{Z}_L'\mathbf{E}\bm{\Omega}^{-1} \label{eq:Derivative}
\end{align}
see for example \cite{Lutkepohl1996} and Appendix \ref{sec:Score}. Evaluated under $H_0$, the score obtained in Eq. \eqref{eq:Derivative} becomes
\begin{equation*}
\frac{\partial \ell_T (\tilde{\bm{\theta}})}{\partial \mathbf{\tilde{D}}_L}\mid_{H_0} = \sum_{t=1}^{T}\mathbf{z}_t\hat{\bm{\varepsilon}}_t'\hat{\bm{\Omega}}^{-1} = \mathbf{Z}_L'\hat{\mathbf{E}}\hat{\bm{\Omega}}^{-1}
\end{equation*}
where $\hat{\mathbf{E}} = \left(\hat{\bm{\varepsilon}}_1, \hat{\bm{\varepsilon}}_2, \ldots, \hat{\bm{\varepsilon}}_T\right)'$, $\hat{\bm{\varepsilon}}_t = \mathbf{y}_t - \hat{\mathbf{D}}_0'\mathbf{x}_t$, and $\hat{\bm{\Omega}} = (1/T)\sum_{t=1}^{T}\hat{\bm{\varepsilon}}_t\hat{\bm{\varepsilon}}_t'$. The matrix $\hat{\mathbf{D}}_0$ is the maximum likelihood (ML) estimator of $\mathbf{D}_0$ under the null hypothesis. The consistency of the ML estimator is guaranteed under the stationarity conditions provided by \cite{Kheifets2020}.

Under Assumptions 1-4, the score vector is asymptotically normally distributed with $n \cdot \text{cd}(\mathbf{Z}_L)$ degrees of freedom, where $\text{cd}(\mathbf{Z}_L)$ is the column dimension of $\mathbf{Z}_L$ \citep{Breusch1980}. As the score is normal and $\mathbf{Z}_L'(\mathbf{I}_T - \mathbf{P_{\mathbf{Z}}})\mathbf{Z}_L$ is positive definite, the vectorised LM test statistic
\begin{equation}\label{eq:LM_lin}
LM_L = \text{vec}\left(\hat{\mathbf{E}}'\mathbf{Z}_L\right)'\left\{\left(\mathbf{Z}_L'(\mathbf{I}_T-\mathbf{P}_{\mathbf{X}})\mathbf{Z}_L\right)\otimes \hat{\bm{\Omega}}\right\}^{-1}\text{vec}\left(\hat{\mathbf{E}}'\mathbf{Z}_L\right)
\end{equation}
has an asymptotic $\chi^2$-distribution with $n \cdot \text{cd}(\mathbf{Z}_L)$ degrees of freedom when the null hypothesis holds.

The statistics in \eqref{eq:LM_lin} can also be written as follows:
\begin{align}
LM_L &= \text{vec}\left(\hat{\mathbf{E}}'\mathbf{Z}_L\right)'\left\{\left(\mathbf{Z}_L'(\mathbf{I}_T-\mathbf{P}_{\mathbf{X}})\mathbf{Z}_L\right)\otimes \hat{\bm{\Omega}}\right\}^{-1}\text{vec}\left(\hat{\mathbf{E}}'\mathbf{Z}_L\right) \nonumber\\
&= \text{vec}\left(\hat{\mathbf{E}}'\mathbf{Z}_L\right)'\left\{\left(\mathbf{Z}_L'(\mathbf{I}_T-\mathbf{P}_{\mathbf{X}})\mathbf{Z}_L\right)^{-1}\otimes \hat{\bm{\Omega}}^{-1}\right\}\text{vec}\left(\hat{\mathbf{E}}'\mathbf{Z}_L\right) \nonumber\\
&= \text{vec}\left(\hat{\mathbf{E}}'\mathbf{Z}_L\right)'\text{vec}\left\{\hat{\bm{\Omega}}^{-1}\hat{\mathbf{E}}'\mathbf{Z}_L\left(\mathbf{Z}_L'(\mathbf{I}_T-\mathbf{P}_{\mathbf{X}})\mathbf{Z}_L\right)^{-1} \right\} \nonumber\\
&= \text{tr}\left\{\mathbf{Z}_L'\hat{\mathbf{E}}\hat{\bm{\Omega}}^{-1}\hat{\mathbf{E}}'\mathbf{Z}_L\left(\mathbf{Z}_L'(\mathbf{I}_T -\mathbf{P}_{\mathbf{X}})\mathbf{Z}_L\right)^{-1}\right\} \nonumber\\
&= \text{tr}\left\{\hat{\bm{\Omega}}^{-1}\hat{\mathbf{E}}'\mathbf{Z}_L\left[\mathbf{Z}_L'(\mathbf{I}_T -\mathbf{P}_{\mathbf{X}})\mathbf{Z}_L\right]^{-1}\mathbf{Z}_L'\hat{\mathbf{E}}\right\}. \label{eq:LM_lin2}
\end{align}
It should be noted that vectorisation and Kronecker products in Eq. \eqref{eq:LM_lin} are avoided in \eqref{eq:LM_lin2}. Then, we have the following result:
\begin{theorem}
The LM test statistic for the null hypothesis, $\text{H}_0 \colon \gamma_i = 0$, $i = 1, \ldots, n$ in Eq. \eqref{eq:VLSTAR3}, or $\text{H}_0 \colon \mathbf{D}_L = \mathbf{0}$ in Eq. \eqref{eq:VLSTAR4}, can be computed as follows:
\begin{equation}\label{eq:LMlinearity}
	LM_L = \text{tr}\left\{\mathbf{\hat{\Omega}}^{-1}\left(\mathbf{Y}-\mathbf{X}\mathbf{\hat{D}}_0\right)'\mathbf{Z}_L\left[\mathbf{Z}_L'\left(\mathbf{I}_T - \mathbf{P}_X\right)\mathbf{Z}_L\right]^{-1}\mathbf{Z}_L'\left(\mathbf{Y}-\mathbf{X}\mathbf{\hat{D}}_0\right)\right\} 
\end{equation}
where $\mathbf{\hat{D}_0}$ is the estimate of $\mathbf{D}_0$. Under the null hypothesis the test statistic has a $\chi^2$-distribution with $L n \left(1+ n p\right)$ degrees of freedom.
\end{theorem}

\begin{proof}
See Appendix \ref{sec:AppendixA}.
\end{proof}

In an asymptotically equivalent way, the test can be performed also in the $TR^2$-form as follows
\begin{enumerate}
	\item Estimate the restricted model under the null hypothesis. Collect the residuals $\mathbf{\hat{\bm{\varepsilon}}}_t = \mathbf{y}_t - \mathbf{X} \hat{\mathbf{D}}_0$. Compute the matrix residual sum of squares $\mathbf{\hat{E}}'\mathbf{\hat{E}}$, where $\mathbf{\hat{E}} = \left[\mathbf{\hat{\bm{\varepsilon}}}_1, \ldots, \mathbf{\hat{\bm{\varepsilon}}}_T\right]'$.
	\item Regress $\mathbf{\hat{E}}$ on $\mathbf{X}$ and $\mathbf{Z}_L$. Collect the residuals, $\mathbf{\hat{\Xi}}$, and form the matrix residual sum of squares $\mathbf{\hat{\Xi}}'\mathbf{\hat{\Xi}}$.
	\item Compute the test statistic
	\begin{align}\label{eq:LMTR1}
		LM_{TR^2} &= T \cdot \text{tr}\left\{\left(\mathbf{\hat{E}}'\mathbf{\hat{E}}\right)^{-1}\left(\mathbf{\hat{E}}'\mathbf{\hat{E}}-\mathbf{\hat{\Xi}}'\mathbf{\hat{\Xi}}\right)\right\} \nonumber\\
		&=  T\left(n - \text{tr}\left\{\left(\mathbf{\hat{E}}'\mathbf{\hat{E}}\right)^{-1}\mathbf{\hat{\Xi}}'\mathbf{\hat{\Xi}}\right\}\right).
	\end{align}
\end{enumerate}
In this setting, the choice of $L$ is somewhat arbitrary. A higher order will increase the column dimension of $\mathbf{Z}_L$, but rejecting the null hypothesis would become easier, since a higher order often increases the the test. On the other hand, a lower order may lead to a test with better size properties, because it uses fewer parameters. As further discussed in \cite{lusate88} for the univariate case, choosing $L=1$ is not a good choice when $s_t = y_{t-d, i}$ for some $1 \leq d \leq p$ and for $i = 1, \ldots, n$, since the LM statistic has only trivial power against this alternative. The problem is typically solved by choosing a third-order Taylor expansion.

A special case of a VLSTAR with a common transition variable is the one with a transition function that is common to all the equations, \textit{i.e.}, $\mathbf{G}_t = g_t(s_t; \gamma, c) \mathbf{I}_n$ with $g(s_t; \gamma, c)$ being a scalar. In such a case, we have that
	\begin{equation*}
		g_t(s_t; \gamma, c) = \sum_{l = 0}^{L}\upsilon_{l} s_t^l + r_t
	\end{equation*}  
	which leads to $\mathbf{\Upsilon}_l = \upsilon_l \mathbf{I}_n$ and $\mathbf{R}_t = r_t \mathbf{I}_n$. Inserting these elements in Eq. \eqref{eq:approxVLSTAR}, the construction of the LM-type statistic remains the same as above.

\section{Determining the number of regimes}\label{sec:Nregimes}

Once rejected the null of linearity, the practitioner should account for the possible presence of some nonlinearity not gathered from a 2-regime model. This means that there may exist an additional nonlinear component that enters the model additively. In this paper, we build upon the additive nonlinearity test introduced by \cite{yate14} which can also be used in a sequential procedure for the detection of the number of regimes. Following the findings by \cite{baipe98}, \cite{Strikholm2006} suggest the use of a sequential testing procedure for additive nonlinearity in the univariate framework. We here extend such a procedure in the multivariate framework to both test additional nonlinearity and specify the number of regimes, $m$. 

If the equations do not share the same transition variable, identifying the number of regimes is not straightforward and a suitable choice would be to select the minimum number of regimes identified in an equation-by-equation test. When a common transition variable is assumed throughout the system, the number of regimes can be identified from the following procedure which mainly extends in the multivariate framework the sequential test proposed by \cite{Strikholm2006}. We further discuss in Section \ref{sec:VTAR} how this procedure can be applied also in the case of a VTAR model as the data-generating process. It should be noticed that the stationarity conditions are available only for a two-regime VLSTAR model \citep{Kheifets2020}, this means that the consistency of the ML estimator, and the stability of the LM test results are guaranteed only for $H_0\colon m = 2$. Consequently, the tests can be only used to suggest the presence of at least three regimes.

The sequential testing procedure can start directly from the case of linearity testing against a 2-regime model. Hence, the first step of the procedure foresees the implementation of the linearity test shown in Eq. \eqref{eq:LMlinearity} to test the null hypothesis of $m = 1$ against $m = 2$. If $\text{H}_0$ is rejected at a given level, $\alpha$, there could exist additive nonlinearity in the model. Therefore, the purpose of the practitioner may be sequentially testing for $m-1$ versus $m$ regimes until a non-rejection.

If we write Eq. \eqref{eq:VLSTAR} for $m = 3$ regimes as follows
\begin{equation}\label{eq:VLSTAR5}
	\mathbf{y}_t = \mathbf{B}_1' \mathbf{x}_t + \mathbf{G}_t^{(1)} \mathbf{B}_2' \mathbf{x}_t + \mathbf{G}_t^{(2)} \mathbf{B}_3' \mathbf{x}_t + \bm{\varepsilon}_t,
\end{equation} 
testing for non-additive nonlinearity equals to test $\text{H}_0 \colon \gamma_{2,i} = 0$, $i = 1, \ldots, n$, against the alternative $H_1 \colon \exists \gamma_{2,i} > 0$. Clearly, the test can be extended to a generic number of $m$ regimes.

As for the linearity test in Section \ref{sec:Linearity}, the alternative model is not identified under the null hypothesis. Once again, Taylor's approximation of $\mathbf{G}_t^{(2)}$ allows us to overcome this problem and obtain a feasible test statistic. Using an $L$-order Taylor approximation, Eq. \eqref{eq:VLSTAR5} becomes
\begin{equation}\label{eq:VLSTAR6}
	\mathbf{y}_t = \mathbf{B}_1' \mathbf{x}_t + \mathbf{G}_t^{(1)} \mathbf{B}_2' \mathbf{x}_t + \left(\sum_{l = 0}^{L}\mathbf{\Upsilon}_{l}^{(2)}s_t^{l} + \mathbf{R}_t^{(2)}\right) \mathbf{B}'_{3}\mathbf{x}_t + \bm{\varepsilon}_t
\end{equation}
where $\mathbf{\Upsilon}_{l}^{(2)}$ is the diagonal matrix of coefficients of the $L$-order Taylor expansion of $g_{i,t}^{(2)}$. As for the linearity test, the null hypothesis implies $\mathbf{\Upsilon}_l^{(2)} = \mathbf{0}$ for $l = 1, \ldots, L$. By reparametrizing, Eq. \eqref{eq:VLSTAR6} can be written as
\begin{align}\label{eq:VLSTAR7}
	\mathbf{y}_t  &= \mathbf{B}_1' \mathbf{x}_t + \mathbf{G}_t^{(1)} \mathbf{B}_2' \mathbf{x}_t + \mathbf{G}_t^{(2)} \mathbf{B}'_{3} \mathbf{x}_t+ \varepsilon_t \nonumber\\
	&= \mathbf{B}_1' \mathbf{x}_t + \mathbf{G}_t^{(1)} \mathbf{B}_2' \mathbf{x}_t + \left(\sum_{l = 0}^{L}\mathbf{\Upsilon}_{l}^{(2)}s_t^{l} + \mathbf{R}_t^{(2)}\right) \mathbf{B}'_{3}\mathbf{x}_t + \bm{\varepsilon}_t \nonumber\\
	&= \left(\mathbf{B}_1' + \mathbf{G}_t^{(1)}\mathbf{B}_2' + \mathbf{\Upsilon}_0^{(2)}\mathbf{B}'_{3}\right)\mathbf{x}_t + \sum_{l = 1}^{L}\mathbf{\Upsilon}_{l}^{(2)}\mathbf{B}'_{3}\mathbf{x}_ts_t^l + \mathbf{R}_t^{(2)}\mathbf{B}'_{3}\mathbf{x}_t + \bm{\varepsilon}_t \nonumber\\
	&= \mathbf{\Psi}_0'\mathbf{x}_t + \sum_{l = 1}^{L}\mathbf{\Psi}_{l}'\mathbf{x}_t s_t^l + \bm{\varepsilon}_t^*
\end{align}
where $\mathbf{\Psi}_0 = \left(\mathbf{B}_1' + \mathbf{G}_t^{(1)}\mathbf{B}_2' + \mathbf{\Upsilon}_0^{(2)}\mathbf{B}'_{3}\right)'$, $\mathbf{\Psi}_l = \left(\mathbf{\Upsilon}_{l}^{(2)}\mathbf{B}'_{3}\right)'$ and $\bm{\varepsilon}_t^* = \mathbf{R}_t^{(2)}\mathbf{B}'_{3}\mathbf{x}_t + \bm{\varepsilon}_t$. 

The null hypothesis in the VAR in Eq. \eqref{eq:VLSTAR7} is $\text{H}_0 \colon \mathbf{\Psi}_1 = \ldots = \mathbf{\Psi}_L = 0$. Let be  $\mathbf{Y} = \left(\mathbf{y}_1', \ldots, \mathbf{y}_T'\right)'$, $\mathbf{X} = \left(\mathbf{x}_1', \ldots, \mathbf{x}_T'\right)'$, $\mathbf{E}$ the $T \times n$ matrix of residuals from Eq. \eqref{eq:VLSTAR7}, and $\mathbf{Z}_L$ as in Eq. \eqref{eq:Z}, model \eqref{eq:VLSTAR7} can be written as
\begin{equation}\label{eq:VLSTAR8}
	\mathbf{Y} = \mathbf{X} \mathbf{\Psi}_0 + \mathbf{Z}_L \mathbf{\Psi}_L + \mathbf{E}.
\end{equation}
Let suppose that 
\begin{equation}\label{eq:Kderiv}
\mathbf{K} = \left[\text{vec}\left(\partial \mathbf{\Psi}_1' \mathbf{B}' \mathbf{x}_1/\partial \bm{\theta}\right)' \ldots \ \text{vec}\left(\partial \mathbf{\Psi}_T' \mathbf{B}' \mathbf{x}_T/\partial \bm{\theta}\right)'\right],
\end{equation}
and that $\mathbf{P}_{\mathbf{K}} = \mathbf{K}(\mathbf{K'K})^{-1}\mathbf{K}'$, the test statistic can be computed similarly to the one in Section \ref{sec:Linearity}, therefore we can state the following result:

\begin{theorem}
	If the estimates of the parameters in Eq. \eqref{eq:VLSTAR8} are consistent, under the null $\text{H}_0 \colon \mathbf{\Psi}_L =\mathbf{ 0}$, the LM test statistic for non-additive nonlinearity
	\begin{equation}\label{eq:Nonadditive}
		LM_L = \text{tr}\left\{\mathbf{\hat{\Omega}}^{-1}\mathbf{\hat{E}}'\mathbf{Z}_L\left[\mathbf{Z}_L'\left(\mathbf{I}_T - \mathbf{P}_K\right)\mathbf{Z}_L\right]^{-1}\mathbf{Z}_L' \mathbf{\hat{E}}\right\}
	\end{equation}
	has an asymptotic $\chi^2$-distribution with $L n(1 + n p)$ degrees of freedom under the Assumptions 1-3 from Section \ref{sec:Linearity}, and under the assumption that $\mathbf{K}'\mathbf{K}$ and $\mathbf{Z}_L'(\mathbf{I}_T - \mathbf{P}_K)\mathbf{Z}_L$ are positive definite matrices.
\end{theorem}
The asymptotic distribution of the LM statistic has the desired null distribution only when $m = 2$ in testing $\mathbf{G}_t^{(m-1)} = (1/2)\mathbf{I}_n$. There are moment conditions for the asymptotic distribution theory to be valid \citep{Eitrheim1996}. In the univariate case, a STAR model with logistic-type transition functions must satisfy the condition $E(\varepsilon_t^8) < \infty$. A sufficient condition in the multivariate case is $E(\varepsilon_{i,t}^8) < \infty$, for $i = 1, \ldots, n$.
To compute $K$ as in Eq. \eqref{eq:Kderiv}, the vector of first-order partial derivatives of $\mathbf{\Psi}_t' \mathbf{B}' \mathbf{x}_t$ is necessary. This can be found in Appendix \ref{sec:firstderiv}. As further discussed in Section \ref{sec:Linearity}, the test can also be performed using the $TR^2$-form in a multi-step regression problem.

The LM test statistics can be, hence, used in a top-down sequential testing procedure that foresees testing the null hypothesis of $\gamma_{m-1,i} = 0$ for a growing number of regimes until non-rejection. Therefore, the number of regimes to be included in the model is the minimum for which the null hypothesis of no-additive nonlinearity cannot be rejected.

\subsection{Applying the sequential procedure for a vector threshold autoregressive model}\label{sec:VTAR}
The tests presented in this study are also valid if the practitioner is analysing a VTAR model. Since the VLSTAR nests the VTAR for $\gamma_{d}$ sufficiently large, the idea is to apply the tests directly on a VLSTAR approximation of the VTAR. This should also solve the drawback of finding the first derivative of the indicator function in the VTAR model or using a bootstrap procedure as proposed for the univariate framework in \cite{Giannerini2024}. Let's suppose a 2-regime VTAR model with a single transition variable and a not-switching error term, defined as follows
	\begin{equation*}
\mathbf{y}_t = \left(\bm{\mu}_1 + \sum_{j = 1}^{p}\mathbf{\Phi}_{1,j}\mathbf{y}_{t-j}\right)\mathbbm{1}\left(s_t \leq c\right) + \left(\bm{\mu}_2 + \sum_{j=1}^{p}\mathbf{\Phi}_{2,j}\mathbf{y}_{t-j}\right)\mathbbm{1}\left(s_t > c\right) + \bm{\varepsilon}_t
	\end{equation*}
which can alternatively be written as
\begin{equation}\label{eq:VTAR}
\mathbf{y}_t = \bm{\mu}_1^* + \sum_{j=1}^{p}\mathbf{\Phi}^*_{1,j}\mathbf{y}_{t-j} + \left(\mu_2^* + \sum_{j=1}^{p}\mathbf{\Phi}_{2,j}^*\mathbf{y}_{t-j}\right)\mathbbm{1}(s_t > c) + \bm{\varepsilon}_t.
\end{equation}
It follows that the indicator function $\mathbbm{1}(\cdot)$ in Eq. \eqref{eq:VTAR} can be approximated by a logistic function where the slope parameter $\gamma$ is fixed and equal to a sufficiently large value. Consequently, the estimated parameters of the approximation are consistent under the same assumptions of the VLSTAR model \cite[see also][]{lusate88}. This means that the aforementioned tests for a VLSTAR model can also be applied when the data-generating process is a VTAR. Moreover, in the sequential procedure for choosing the number of regimes, the estimates of the threshold parameters are super-consistent and assuming them known makes it easier to test $m = m_0-1$ against $m = m_0$, where $m_0 \geq 2$. 

Consider a single lag three-regime VTAR model (where, for simplicity, we imposed $\bm{\mu}_1 = \bm{\mu}_2 = \bm{\mu}_3 = \mathbf{0}$)
\begin{equation*}
\mathbf{y}_t = \left(\mathbf{\Phi}_1\mathbf{y}_{t-1}\right) \mathbbm{1}(s_t \leq c_1) + \left(\mathbf{\Phi}_2\mathbf{y}_{t-1}\right) \mathbbm{1}(c_1 <s_t\leq c_2) + \left(\mathbf{\Phi}_3\mathbf{y}_{t-1}\right) \mathbbm{1}(s_t >c_2) + \bm{\varepsilon}_t
\end{equation*}
which can be reformulated as
\begin{equation}\label{eq:VTARref}
\mathbf{y}_t = \mathbf{\Phi}_1^*\mathbf{y}_{t-1} + \left(\mathbf{\Phi}_2^*\mathbf{y}_{t-1}\right)\mathbbm{1}(s_t > c_1) + \left(\mathbf{\Phi}_3^*\mathbf{y}_{t-1}\right)\mathbbm{1}(s_t > c_2) + \bm{\varepsilon}_t
\end{equation}
with $c_1 <c_2$. The sequential test procedure can be summarized in the following steps:
\begin{enumerate}
\item Set $\mathbf{\Phi}_3^* = \mathbf{0}$ in Eq. \eqref{eq:VTARref} and approximate the indicator function $\mathbbm{1}(s_t >c_1)$ with $\mathbf{G}_t^{(1)}(s_t;\bm{\gamma}_1, \mathbf{c}_1) = g_t^{(1)}(s_t; \gamma_1, c_1)\mathbf{I}_n$, where $g_t^{(1)}$ is a standard logistic function, such that
\begin{equation*}
	\mathbf{y}_t = \mathbf{\Phi}_1^*\mathbf{y}_{t-1}  + \left(\mathbf{\Phi}_2^*\mathbf{y}_{t-1}\right)\mathbf{G}_t^{(1)}(s_t;\bm{\gamma}_1, \mathbf{c}_1) + \bm{\varepsilon}_t.
\end{equation*}
Then, apply the linearity test presented in Section \ref{sec:Linearity} for a VLSTAR model, imposing $H_0 \colon \gamma_1 = 0$.
\item If the null hypothesis is rejected at a given significance level $\alpha$, estimate the coefficients in \eqref{eq:VTARref} model imposing $\mathbf{\Phi}_3^*=\mathbf{0}$. As further detailed in \cite{Chan1993} and \cite{Gonzalo2002}, the threshold estimator, $\hat{c}_1$, is super consistent.
\item Use the super consistent estimator, $\hat{c}_1$, in Eq. \eqref{eq:VTARref} and test the linearity of the following model
\begin{equation*}
\mathbf{y}_t = \mathbf{\Phi}_1^*\mathbf{y}_{t-1} + \left(\mathbf{\Phi}_2^*\mathbf{y}_{t-1}\right)\mathbbm{1}(s_t > \hat{c}_1) + \left(\mathbf{\Phi}_3^*\mathbf{y}_{t-1}\right)\mathbf{G}_t^{(2)}(s_t; \bm{\gamma}_2, \mathbf{c}_2) + \bm{\varepsilon}_t
\end{equation*}
where $\mathbf{G}_t^{(2)}(s_t;\bm{\gamma}_2, \mathbf{c}_2) = g_t^{(2)}(s_t; \gamma_2, c_2)\mathbf{I}_n$.
\item If the null hypothesis, $H_0 \colon \gamma_2 = 0$, is rejected at a significance level, estimate the parameters in Eq. \eqref{eq:VTARref} and use the super consistent estimate, $\hat{c}_2$, to test the linearity of \eqref{eq:VTARref} against a four-regime model.
\item Continue the procedure until a non-rejection.
\end{enumerate}
This, indeed, extends in the multivariate the sequential procedure for the definition of the number of thresholds proposed by \cite{Strikholm2006}.

\section{Simulation Study}\label{sec:Simulation}

In vector models, the standard LM-type tests can be strongly oversized when the null hypothesis foresees the estimation of a large set of parameters and when the size of the sample is not large. In practice, the nominal size of the test tends to overestimate the true probability of type I error in finite samples, see also \cite{Honda1988}. As in \cite{Laitinen1978} and \cite{Meisner1979}, to overcome this limitation, we use a Bartlett-type correction that allows to rescale the degrees of freedom of the test and apply an $F$-statistic. In a Monte Carlo simulation study conducted by \cite{Bera1981}, the authors show that this correction is able to correct the oversize of LM. 

The Laitinen-Meisner correction consists of a degree of freedom rescaling of the form $(nT - S)/(W \times nT)$, where $n$ and $T$ are defined as before, $S$ is the number of parameters, and $W$ is the number of restrictions, see \cite{Laitinen1978} and \cite{Meisner1979}. The $F$-type LM test statistic, or rescaled LM test statistic, can be computed as
\begin{equation}\label{eq:LMadj}
	LM_L^{\text{resc}} = LM_L \cdot \frac{nT - S}{W \times nT}
\end{equation}
and follows an $F\left(W, nT - S\right)$ distribution.

We carry out some simulation experiments for the finite-sample performance of the test procedure. Specifically, we investigate the empirical size and the power of the test, and we also report the selection frequencies for the sequential procedure. We also consider a Wilks' lambda test statistic based on Wilks' $\Lambda$-distribution \citep{anderson2003}. In Appendix \ref{sec:Ftests}, we show that Wilks' $\Lambda$ is applicable in our testing situation and how the test is performed in our framework.

To compute the empirical sizes of the sequential test procedure, we generate $1000$ replications from the model specified in Eq. \eqref{eq:VLSTAR5}. We do not include any explanatory variable and we select a single lag for the simulation of $n = 3$ dependent variables, such that:
\begin{equation}\label{eq:Simulation}
	\mathbf{y}_t = \mathbf{B}_{1}'\mathbf{y}_{t-1} + \mathbf{G}_t^{(1)}\mathbf{B}_{2}'\mathbf{y}_{t-1} + \mathbf{G}_t^{(2)}\mathbf{B}_{3}'\mathbf{y}_{t-1} + \ldots +\mathbf{G}_t^{(m-1)} \mathbf{B}_{m1}'\mathbf{y}_{t-1} + \bm{\varepsilon}_t.
\end{equation}
where $\bm{\varepsilon}_t \sim \mathcal{N}\left(0, \mathbf{I}_n\right)$, thus we are supposing uncorrelated errors.

For each realization, we estimate the VLSTAR model and compute the residuals matrix. Since the VLSTAR model is estimated numerically, relatively large samples are required for a reasonable estimation accuracy, therefore we choose $T = 400, 600, 1000$.

In a first attempt, we assess the empirical size of the test procedure by simulating $\mathbf{y}_t$ from Eq. \eqref{eq:Simulation} with $m = 2$ and we test the null hypothesis of $m = 2$ against $m = 3$. The results are reported in Table \ref{table:Empiricalm2}. 
The DGP relies on a parameter matrix $\mathbf{B}_1$ with 0.1 entries and diagonal values $\rho_i$, for $i = 1, \ldots, n$, sampled either from a uniform distribution $\mathcal{U}(0.3, 0.5)$ or from $\mathcal{U}(0.5,0.8)$, we also set $\mathbf{B}_2 = -\mathbf{B}_1$ 
and we use $c = \gamma = 2$. Finally, the common transition variable is generated by an exogenous first-order AR process, such that
\begin{equation}\label{eq:st}
	s_t = 0.95 s_{t-1} + \eta_t
\end{equation}
where $\eta_t \sim \mathcal{N}(0, 1)$. As in \cite{Strikholm2006}, we use three different nominal sizes, $\alpha = 0.10, 0.05, 0.01$. To compute the $TR^2$-form of the test statistics, we use a third-order Taylor expansion, therefore $L = 3$.

Analysing the empirical sizes of the test in Table \ref{table:Empiricalm2}, it can be noticed that these are close to the nominal values in all of the three tests. When simulating $T = 1000$ observations, the empirical size of the LM test statistics (first three columns) and the Wilks' statistics (last three columns) almost coincides with the nominal size for $\alpha = 0.10$ and $\alpha = 0.05$, while the rescaled LM has an empirical size close to the real one for $T = 400$ and $T=600$. The persistence level of the dependent variables seems to be relevant, since a diagonal value $\rho_i \sim \mathcal{U}(0.5, 0.8)$ leads to slightly divergent results in terms of empirical size.

We further evaluate the empirical power of the tests in finite-size samples by applying it to simulated data from Eq. \eqref{eq:Simulation} with $m=3$ regimes (with $\mathbf{B}_3 = -0.7\cdot \mathbf{I}_n$, $\gamma_2= 2$ and $c_2 = 4$). The power of the test is then calculated by testing the null hypothesis of $m = 2$ regimes. The results in Table \ref{table:Power} suggest that the LM test has a good empirical power and that this increases with sample size. This is not entirely surprising since the VLSTAR model, especially with higher-order regimes, requires the estimation of a large set of parameters  to obtain reasonable estimation accuracy.

\begin{table}[H]
	\centering
		\caption{\textbf{Empirical size of additive nonlinearity test. The empirical size is in per cent based on 1000 replications from model \eqref{eq:Simulation} with $m = 2$ regimes and $n = 3$}}\label{table:Empiricalm2}
		\begin{tabular}{lrrr@{\hskip 0.3cm}rrr@{\hskip 0.3cm}rrr}
			\toprule
& \multicolumn{3}{c}{$LM_3$} & \multicolumn{3}{c}{$LM_3^{\text{resc}}$} & \multicolumn{3}{c}{Wilks}\\
\cmidrule(lr){2-4} \cmidrule(lr){5-7} \cmidrule(lr){8-10}\\
& 10\%  & 5\% & 1\% & 10\%& 5\% &    1\%& 10\%     & 5\%    &    1\%\\
\midrule

$T = 400$                                    &      &      &      &      &      &      \\[2pt]
$\rho_i \sim \mathcal{U}(0.3, 0.5)$&  14.1  & 7.6 & 1.6 & 9.9 & 5.1 & 0.7 & 14.0 & 7.8 & 1.6\\
$\rho_i \sim \mathcal{U}(0.5, 0.8)$&  15.1  & 8.6 & 1.7 & 11.2& 5.9 & 1.2 & 14.9 & 8.8 & 1.8\\[5pt]

$T = 600$                                     &      &      &      &      &       &      \\[2pt]
$\rho_i \sim \mathcal{U}(0.3, 0.5)$& 10.4  & 6.3 & 1.6 & 8.8 & 5.0 & 1.1 & 10.5 & 6.3 & 1.6\\
$\rho_i \sim \mathcal{U}(0.5, 0.8)$& 12.4  & 7.1 & 2.4 & 10.2& 5.4 & 2.0 & 12.3 & 6.9 & 2.4\\[5pt]

$T = 1000$                                     &      &      &      &      &       &      \\[2pt]
$\rho_i \sim \mathcal{U}(0.3, 0.5)$& 9.1 & 4.9 & 1.5 & 7.8 & 4.4 & 0.9 & 9.0 & 4.9 & 1.5\\
$\rho_i \sim \mathcal{U}(0.5, 0.8)$& 8.3 & 4.1 & 0.7 & 7.4 & 3.6 & 0.6 & 8.3 & 4.2 & 0.7\\
			\bottomrule  
			\multicolumn{10}{p{0.7\textwidth}}{Note: Data are simulated from a model with $m = 2$ regimes and we test the null of $\text{H}_0 \colon m = 2$ against $H_1 \colon m = 3$. We simulate using different values of $\mathbf{B}_1$ and for $\mathbf{B}_2 = - \mathbf{B}_1$.}  
		\end{tabular}
\end{table}

\begin{table}[H]
	\centering
	\caption{\textbf{Power of the additive nonlinearity test at different sample sizes based on 1000 replications from model \eqref{eq:Simulation} with $m = 3$ regimes and $n = 3$, testing the null $\text{H}_0 \colon m = 2$}}\label{table:Power}
	\begin{tabularx}{0.67\textwidth}{lYYY}
		\toprule
		& $LM_3$ & $LM_3^{\text{resc}}$  & Wilks\\
\midrule
$T = 400$ & \\[2pt]
$\rho_i \sim \mathcal{U}(0.3, 0.5)$& 85.7 & 83.0 & 84.8\\
$\rho_i \sim \mathcal{U}(0.5, 0.8)$& 89.5 & 86.2 & 89.0\\[5pt]
$T = 600$ & \\[2pt]
$\rho_i \sim \mathcal{U}(0.3, 0.5)$& 89.4 & 87.5 & 88.6\\
$\rho_i \sim \mathcal{U}(0.5, 0.8)$& 92.7 & 91.8 & 92.7\\[5pt]
$T = 1000$ & \\[2pt]
$\rho_i \sim \mathcal{U}(0.3, 0.5)$& 93.8 & 93.2 & 93.6\\
$\rho_i \sim \mathcal{U}(0.5, 0.8)$& 96.1 & 95.3 & 96.0\\
\bottomrule
\multicolumn{4}{p{0.65\textwidth}}{Note: Data are simulated from a model with $m = 3$ regimes and we test the null of $\text{H}_0 \colon m = 2$ against $H_1 \colon m = 3$. We simulate using different values of $\mathbf{B}_1$, for $\mathbf{B}_2 = -\mathbf{B}_1$ and $\mathbf{B}_3 = -0.7 \cdot \mathbf{I}_n$, with $\gamma_1= \gamma_2 = 2$, $c_1 = 2$ and $c_2 = 4$. The nominal significance level is set equal to 0.05.}  
	\end{tabularx}
\end{table}

In order to evaluate the regime choice of the procedures introduced in this paper, we also report in Table \ref{table:SelFreqm2} the selection frequencies when the data are simulated, as before, from model \eqref{eq:Simulation} with $m = 2$, with the same characteristic presented above. The procedures start with a linearity test against a two-regime model, then foresee testing $m=2$ vs $m=3$ regimes and continue until a non-rejection.

For any sample size and persistence level, the selection frequencies indicate a high tendency to identify two regimes correctly. As expected, increasing the sample size enhances the accuracy of the tests in identifying the correct number of regimes. In fact, for both $\rho_i \sim \mathcal{U}(0.3, 0.5)$ and $\rho_i \sim \mathcal{U}(0.5, 0.8)$, the selection frequencies of $\hat{m} = 2$ are all above 97\% for all the tests when the significance level is $\alpha = 0.01$ and $T=1000$.

\begin{table}[H]
	\centering
	\caption{\textbf{Selection frequencies for a VLSTAR model. The frequency is in per cent based on 1000 replications from model \eqref{eq:Simulation} with $m = 2$ regimes and $n = 3$}}\label{table:SelFreqm2}
	\begin{tabular}{lrrr@{\hskip 0.2cm}rrr@{\hskip 0.2cm}rrr}
		\toprule
		& \multicolumn{3}{c}{$LM_3$} & \multicolumn{3}{c}{$LM_3^{\text{resc}}$} & \multicolumn{3}{c}{Wilks}\\
		\cmidrule(lr){2-4} \cmidrule(lr){5-7} \cmidrule(lr){8-10}\\
		& $\hat{m}=1$  & $\hat{m}=2$ & $\hat{m}\geq3$ & $\hat{m}=1$  & $\hat{m}=2$ & $\hat{m}\geq3$& $\hat{m}=1$  & $\hat{m}=2$ & $\hat{m}\geq3$\\
		\midrule
		
		\multicolumn{3}{l}{$T = 400, \rho_i \sim \mathcal{U}(0.3, 0.5)$} \\[2pt]
$\alpha = 0.10$& 0.0 & \underline{85.9} & 14.1  & 0.0 &\underline{90.1} & 9.9 & 11.7 & \underline{74.8} & 13.5\\
$\alpha = 0.05$& 0.1 & \underline{92.3} & 7.6  & 0.1 & \underline{94.8} & 5.1 & 17.8 & \underline{74.7} & 7.5\\
$\alpha = 0.01$& 0.4 & \underline{98.0} & 1.6  & 0.7 & \underline{98.6} & 0.7 & 34.9 & \underline{63.8} & 1.3\\[5pt]

$T = 400, \rho_i \sim \mathcal{U}(0.5, 0.8)$ \\[2pt]
$\alpha = 0.10$& 0.0 & \underline{84.9} & 15.1 & 0.0  & \underline{88.8} & 11.2 & 0.5 & \underline{84.6} & 14.9\\
$\alpha = 0.05$& 0.0 & \underline{91.4} & 8.6  & 0.0 & \underline{94.1} & 5.9  & 1.1 & \underline{90.1} & 8.8\\
$\alpha = 0.01$& 0.0 & \underline{98.3} & 1.7  & 0.0 & \underline{98.8} & 1.2  & 1.9 & \underline{96.3} & 1.8\\[5pt]	
		
$T = 600, \rho_i \sim \mathcal{U}(0.3, 0.5)$ \\[2pt]
$\alpha = 0.10$& 0.0 & \underline{89.6} & 10.4 & 0.0 & \underline{91.2} & 8.8 & 3.0 & \underline{86.7} & 10.3\\
$\alpha = 0.05$& 0.0 & \underline{93.7} & 6.3  & 0.0 & \underline{95.0} & 5.0 & 4.7 & \underline{89.0} & 6.3\\
$\alpha = 0.01$& 0.0 & \underline{98.4} & 1.6  & 0.0 & \underline{98.9} & 1.1 & 10.1 & \underline{88.3} & 1.6\\[5pt]

$T = 600, \rho_i \sim \mathcal{U}(0.5, 0.8)$ \\[2pt]
$\alpha = 0.10$& 0.0 & \underline{87.6} & 12.4  & 0.0 & \underline{89.8}& 10.2 & 0.1 & \underline{87.6} & 12.3\\
$\alpha = 0.05$& 0.0 & \underline{92.9} & 7.1  & 0.0 & \underline{94.6} & 5.4  & 0.1 & \underline{93.0} & 6.9\\
$\alpha = 0.01$& 0.0 & \underline{97.6} & 2.4  & 0.0 & \underline{98.0} & 2.0  & 0.5 & \underline{97.1} & 2.4\\[5pt]
		
$T = 1000, \rho_i \sim \mathcal{U}(0.3, 0.5)$ \\[2pt]
$\alpha = 0.10$& 0.0 & \underline{90.9} & 9.1 & 0.0 & \underline{92.2} & 7.8 & 0.0 & \underline{91.0} & 9.0\\
$\alpha = 0.05$& 0.0 & \underline{95.1} & 4.9 & 0.0 & \underline{95.6} & 4.4 & 0.0 & \underline{95.1} & 4.9\\
$\alpha = 0.01$& 0.0 & \underline{98.5} & 1.5 & 0.0 & \underline{99.1} & 0.9 & 0.2 & \underline{98.3} & 1.5\\[5pt]

$T = 1000, \rho_i \sim \mathcal{U}(0.5, 0.8)$ \\[2pt]
$\alpha = 0.10$& 0.0 & \underline{91.7} & 8.3  & 0.0 & \underline{92.6} & 7.4  & 0.0 & \underline{91.7} & 8.3\\
$\alpha = 0.05$& 0.0 & \underline{95.9} & 4.1  & 0.0 & \underline{96.4} & 3.6  & 0.0 & \underline{95.8} & 4.2\\
$\alpha = 0.01$& 0.0 & \underline{99.3} & 0.7  & 0.0 & \underline{99.4} & 0.6  & 0.0 & \underline{99.3} & 0.7\\[5pt]
		\bottomrule  
		\multicolumn{10}{p{\textwidth}}{Note: Data are simulated from a model with $m = 2$ regimes. We simulate using different values of $\mathbf{B}_1$ and for $\mathbf{B}_2 = - \mathbf{B}_1$. Underlined values denote the identified number of regimes.}  
	\end{tabular}
\end{table}

We then compute the empirical size of the test and the selection frequencies when data are simulated from a VTAR model with $2$ regimes, specified as follows
\begin{equation}\label{eq:TVAR2}
\mathbf{y}_t = \left(\mathbf{\Phi}_1 \mathbf{y}_{t-1} + \bm{\varepsilon}_{1,t}\right)\mathbbm{1}(s_t<c) + \left(\mathbf{\Phi}_2 \mathbf{y}_{t-1} + \bm{\varepsilon}_{2,t}\right)\mathbbm{1}(s_t \geq c)
\end{equation}
where $s_t$ is generated from an autoregressive model as in Eq. \eqref{eq:st}. We let $\mathbf{\Phi}_1$ vary as $\mathbf{B}_1$ before, we also set $\mathbf{\Phi}_2 = -\mathbf{\Phi}_1$, the threshold for $s_t$ is $c = 2$, while $\bm{\varepsilon}_{1,t}, \bm{\varepsilon}_{2,t} \sim \mathcal{N}(0, \mathbf{I}_n)$.

In Table \ref{table:TVAR}, we evaluate the empirical sizes of the test applied to the data simulated from a VTAR specified as in Eq. \eqref{eq:TVAR2}. With the exception of the rescaled LM, all the tests tend to slightly over-reject the null at any significance level. Nevertheless, the difference between the empirical and the nominal size is negligible when the persistence changes.

When analysing the empirical power in Table \ref{table:Power2} (simulating from a three-regime model with $\bm{\Phi}_3 = -0.7\cdot \mathbf{I}_n$ and $c_2 = 4$), it can be observed that the power of the test is generally high and that improves with larger sample sizes, underlining, again, the importance of sample size in detecting additional regimes. Although the empirical power is above 75\%, Wilks' $\Lambda$ tends to perform worse than the other two specifications. 

Table \ref{table:SelFreqVTARm2} complements these findings by showing the selection frequencies for a VTAR model when the real DGP has $m=2$. The results are consistent with those of the VLSTAR model, confirming that the sequential procedure is equally applicable and reliable for VTAR models. As with the VLSTAR model, the accuracy of the test procedure increases with sample size and significance levels.

\begin{table}[H]
	\centering
	\caption{\textbf{Empirical size of additive nonlinearity test for a VTAR model. The empirical size is in per cent based on 1000 replications from model \eqref{eq:TVAR2} with $m = 2$ regimes and $n = 3$}}\label{table:TVAR}
	\begin{tabular}{lrrr@{\hskip 0.3cm}rrr@{\hskip 0.3cm}rrr}
		\toprule
		& \multicolumn{3}{c}{$LM_3$} & \multicolumn{3}{c}{$LM_3^{\text{resc}}$}   & \multicolumn{3}{c}{Wilks}\\
		\cmidrule(lr){2-4} \cmidrule(lr){5-7}  \cmidrule(lr){8-10} \\
		& 10\%  & 5\%   & 1\%   & 10\%  & 5\%   &    1\%  & 10\%  & 5\%   &    1\%\\
		\midrule

$T = 400$             &      &      &      &      &      &      \\
$\rho_i \sim \mathcal{U}(0.3, 0.5)$& 12.6 & 6.9 & 0.7 & 9.2 & 4.3 & 0.3 & 12.2 & 6.8 & 0.9 \\ 
$\rho_i \sim \mathcal{U}(0.5, 0.8)$& 15.6 & 8.5 & 2.6 & 10.4 & 5.8 & 1.6 & 15.1 & 8.6 & 2.6\\ [5pt]

$T = 600$             &      &      &     &      &      &      \\
$\rho_i \sim \mathcal{U}(0.3, 0.5)$&  10.6 & 5.6 & 0.7 &  8.2 & 3.3 & 0.6 & 10.5 & 5.5 & 0.7\\ 
$\rho_i \sim \mathcal{U}(0.5, 0.8)$&  13.5 & 7.6 & 1.8 & 11.1 & 5.4 & 1.4 & 13.4 & 7.5 & 1.9\\[5pt] 

$T = 1000$             &      &      &     &      &      &      \\
$\rho_i \sim \mathcal{U}(0.3, 0.5)$&  11.2 & 5.9 & 1.6 &  9.4 & 4.9 & 1.4 & 11.1 & 5.9 & 1.6\\ 
$\rho_i \sim \mathcal{U}(0.5, 0.8)$&  13.1 & 6.2 & 1.5 & 10.6 & 5.2 & 1.1 & 13.0 & 6.1 & 1.5\\
		\bottomrule  
		\multicolumn{10}{p{0.7\textwidth}}{Note: Data are simulated from a VTAR model with $m = 2$ regimes and we test the null of $\text{H}_0 \colon m = 2$ against $H_1 \colon m = 3$. We simulate using different values of $\bm{\Phi}_1$ and for $\bm{\Phi}_2 = -\bm{\Phi}_1$.}  
	\end{tabular}
\end{table}

\begin{table}[H]
	\centering
	\caption{\textbf{Power of the additive nonlinearity test at different sample sizes based on 1000 replications from model \eqref{eq:TVAR2} with $m = 3$ regimes and $n = 3$, testing the null $\text{H}_0 \colon m = 2$}}\label{table:Power2}
	\begin{tabularx}{0.67\textwidth}{lYYY}
		\toprule
		& $LM_3$ & $LM_3^{\text{resc}}$  & Wilks\\
		\midrule
		$T = 400$ & \\[2pt]
		$\rho_i \sim \mathcal{U}(0.3, 0.5)$& 93.2 & 91.6 & 76.6\\
		$\rho_i \sim \mathcal{U}(0.5, 0.8)$& 94.6 & 92.6 & 80.3\\[5pt]
		$T = 600$ & \\[2pt]
		$\rho_i \sim \mathcal{U}(0.3, 0.5)$& 94.7 & 94.0 & 80.6\\
		$\rho_i \sim \mathcal{U}(0.5, 0.8)$& 95.7 & 95.3 & 84.2\\[5pt]
		$T = 1000$ & \\[2pt]
		$\rho_i \sim \mathcal{U}(0.3, 0.5)$& 95.5 & 95.3 & 86.7\\
		$\rho_i \sim \mathcal{U}(0.5, 0.8)$& 97.2 & 96.9 & 85.1\\
		\bottomrule
		\multicolumn{4}{p{0.65\textwidth}}{Note: Data are simulated from a model with $m = 3$ regimes and we test the null of $\text{H}_0 \colon m = 2$ against $H_1 \colon m = 3$. We simulate using different values of $\bm{\Phi}_1$, for $\bm{\Phi}_2 = -\bm{\Phi}_1$ and $\bm{\Phi}_3 = -0.7 \cdot \mathbf{I}_n$, with $c_1 = 2$ and $c_2 = 4$. The nominal significance level is set equal to 0.05.}  
	\end{tabularx}
\end{table}

\begin{table}[H]
	\centering
	\caption{\textbf{Selection frequencies for a VTAR model. The frequency is in per cent based on 1000 replications from model \eqref{eq:TVAR2} with $m = 2$ regimes and $n = 3$}}\label{table:SelFreqVTARm2}
	\begin{tabular}{lrrr@{\hskip 0.2cm}rrr@{\hskip 0.2cm}rrr}
		\toprule
		& \multicolumn{3}{c}{$LM_3$} & \multicolumn{3}{c}{$LM_3^{\text{resc}}$} & \multicolumn{3}{c}{Wilks}\\
		\cmidrule(lr){2-4} \cmidrule(lr){5-7} \cmidrule(lr){8-10}\\
		& $\hat{m}=1$  & $\hat{m}=2$ & $\hat{m}\geq3$ & $\hat{m}=1$  & $\hat{m}=2$ & $\hat{m}\geq3$& $\hat{m}=1$  & $\hat{m}=2$ & $\hat{m}\geq3$\\
		\midrule
		
		\multicolumn{3}{l}{$T = 400, \rho_i \sim \mathcal{U}(0.3, 0.5)$} \\[2pt]
$\alpha = 0.10$& 0.0 & \underline{87.4} & 12.6 & 0.0 & \underline{90.8} & 9.2 & 0.0 & \underline{87.8} & 12.2\\
$\alpha = 0.05$& 0.0 & \underline{93.1} & 6.9  & 0.0 & \underline{95.7} & 4.3 & 0.0 & \underline{93.2} & 6.8\\
$\alpha = 0.01$& 0.0 & \underline{99.3} & 0.7  & 0.0 & \underline{99.7} & 0.3 & 0.0 & \underline{99.1} & 0.9\\[5pt]

$T = 400, \rho_i \sim \mathcal{U}( 0.8)$ \\[2pt]
$\alpha = 0.10$& 0.0 & \underline{84.4} & 15.6 & 0.0 & \underline{89.6} & 10.4 & 0.0 & \underline{84.9} & 15.1\\
$\alpha = 0.05$& 0.0 & \underline{91.5} & 8.5  & 0.0 & \underline{94.2} & 5.8  & 0.0 & \underline{91.4} & 8.6\\
$\alpha = 0.01$& 0.0 & \underline{97.4} & 2.6  & 0.0 & \underline{98.4} & 1.6  & 0.0 & \underline{97.4} & 2.6\\[5pt]	
		
$T = 600, \rho_i \sim \mathcal{U}(0.3, 0.5)$ \\[2pt]
$\alpha = 0.10$& 0.0 & \underline{89.4} & 10.6 & 0.0 & \underline{91.8} & 8.2 & 0.0 & \underline{89.5} & 10.5\\
$\alpha = 0.05$& 0.0 & \underline{94.4} & 5.6  & 0.0 & \underline{96.7} & 3.3 & 0.0 & \underline{94.5} & 5.5\\
$\alpha = 0.01$& 0.0 & \underline{99.3} & 0.7  & 0.0 & \underline{99.4} & 0.6 & 0.1 & \underline{99.2} & 0.7\\[5pt]

$T = 600, \rho_i \sim \mathcal{U}(0.5, 0.8)$ \\[2pt]
$\alpha = 0.10$& 0.0 & \underline{86.5} & 13.5 & 0.0 & \underline{88.9} & 11.1 & 0.0 & \underline{86.6} & 13.4\\
$\alpha = 0.05$& 0.0 & \underline{92.4} & 7.6  & 0.0 & \underline{94.6} & 5.4  & 0.0 & \underline{92.5} & 7.5\\
$\alpha = 0.01$& 0.0 & \underline{98.2} & 1.8  & 0.0 & \underline{98.6} & 1.4  & 0.0 & \underline{98.1} & 1.9\\[5pt]
		
$T = 1000, \rho_i \sim \mathcal{U}(0.3, 0.5)$ \\[2pt]
$\alpha = 0.10$& 0.0 & \underline{88.8} &11.2 & 0.0 & \underline{90.6} & 9.4 & 0.0 & \underline{88.9} & 11.1\\
$\alpha = 0.05$& 0.0 & \underline{94.1} & 5.9 & 0.0 & \underline{95.1} & 4.9 & 0.0 & \underline{94.1} & 5.9\\
$\alpha = 0.01$& 0.0 & \underline{98.4} & 1.6 & 0.0 & \underline{98.6} & 1.4 & 0.0 & \underline{98.4} & 1.6\\[5pt]

$T = 1000, \rho_i \sim \mathcal{U}(0.5, 0.8)$ \\[2pt]
$\alpha = 0.10$& 0.0 & \underline{86.9} & 13.1 & 0.0 & \underline{89.4} & 10.6 & 0.0 & \underline{87.0} & 13.0\\
$\alpha = 0.05$& 0.0 & \underline{93.8} & 6.2  & 0.0 & \underline{94.8} & 5.2 & 0.0 & \underline{93.9} & 6.1\\
$\alpha = 0.01$& 0.0 & \underline{98.5} & 1.5  & 0.0 & \underline{98.9} & 1.1 & 0.0 & \underline{98.5} & 1.5\\[5pt]
\bottomrule  
\multicolumn{10}{p{\textwidth}}{Note: Data are simulated from a model with $m = 2$ regimes. We simulate using different values of $\bm{\Phi}_1$ and for $\bm{\Phi}_2 = -\bm{\Phi}_1$. Underlined values denote the identified number of regimes.}  
	\end{tabular}
\end{table}

As a robustness check, we also observe what happens when the number of dependent variables increases (we use $n=5$). We simulate only with $\rho_i \sim \mathcal{U}(0.3, 0.5)$, because the DGPs are not stationary under $\rho_i \sim \mathcal{U}(0.5, 0.8)$ \cite[stationarity was assessed through the method proposed in][]{Kheifets2020}.

Table \ref{table:Empiricalm2n5} assesses the empirical size of the additive nonlinearity test in models with $m=2$ regimes and $n=5$ variables. The results slightly diverge from what observed with $n=3$, since the test statistics generally exhibit sizes lower to the nominal levels for the VLSTAR model and higher for the VTAR model. Nevertheless, the empirical sizes for both VLSTAR and VTAR models are overall not too far from the expected values.

The findings are more encouraging when one analyses the empirical powers in Table \ref{table:Powern5}. The power of all three tests is always higher than 0.90 and improves with larger sample sizes. For instance, with $T=1000$, the empirical power is nearly perfect, reflecting the tests' ability to correctly identify additional regimes in large samples.

As for the case of $n=3$, when the number of dependent variables is equal to 5, the procedure is capable of correctly identifying the real number of regimes for any model, sample size and level of persistence. In fact, the number of selected regimes is always equal to 2, with percentage values ranging from 76.7 to 99.8.

\begin{table}[H]
	\centering
	\caption{\textbf{Empirical size of additive nonlinearity test. The empirical size is in per cent based on 1000 replications from a model with $m = 2$ regimes and $n = 5$}}\label{table:Empiricalm2n5}
	\begin{tabular}{lrrr@{\hskip 0.3cm}rrr@{\hskip 0.3cm}rrr}
		\toprule
		& \multicolumn{3}{c}{$LM_3$} & \multicolumn{3}{c}{$LM_3^{\text{resc}}$} & \multicolumn{3}{c}{Wilks}\\
		\cmidrule(lr){2-4} \cmidrule(lr){5-7} \cmidrule(lr){8-10}\\
		& 10\%  & 5\% & 1\% & 10\%& 5\% &    1\%& 10\%     & 5\%    &    1\%\\
		\midrule
&\multicolumn{9}{c}{Panel A: VLSTAR}\\
		
		$T = 400$                                    &      &      &      &      &      &      \\[2pt]
		$\rho_i \sim \mathcal{U}(0.3, 0.5)$&  7.2 & 3.7 & 1.4 & 3.3 & 1.6 & 0.2 & 7.0 & 3.7 & 1.4\\[5pt]

		$T = 600$                                     &      &      &      &      &       &      \\[2pt]
		$\rho_i \sim \mathcal{U}(0.3, 0.5)$& 7.8 & 5.0 & 1.8 & 5.0 & 3.0 & 1.4 & 7.6 & 5.0 & 1.6\\[5pt]
		
		$T = 1000$                                     &      &      &      &      &       &      \\[2pt]
		$\rho_i \sim \mathcal{U}(0.3, 0.5)$& 3.2 & 1.4 & 0.4 & 2.4 & 1.0 & 0.2 & 3.2 & 1.2 & 0.4 \\[2pt]
		\hline
&\multicolumn{9}{c}{Panel B: VTAR}\\

$T = 400$                                    &      &      &      &      &      &      \\[2pt]
$\rho_i \sim \mathcal{U}(0.3, 0.5)$&  10.3  & 7.4 & 3.7 & 5.9 & 3.9 & 1.2 & 10.3 & 7.2 & 3.7\\[5pt]

$T = 600$                                     &      &      &      &      &       &      \\[2pt]
$\rho_i \sim \mathcal{U}(0.3, 0.5)$& 16.3 & 9.8 & 3.8 & 16.9 & 9.7 & 2.9 & 16.7 & 10.7 & 2.7\\[5pt]

$T = 1000$                                     &      &      &      &      &       &      \\[2pt]
$\rho_i \sim \mathcal{U}(0.3, 0.5)$& 15.6 & 9.5 & 3.4 & 12.0 & 7.1 & 2.1 & 15.1 & 9.4 & 3.5\\
		\bottomrule  
		\multicolumn{10}{p{0.7\textwidth}}{Note: Data are simulated from a model with $m = 2$ regimes and we test the null of $\text{H}_0 \colon m = 2$ against $H_1 \colon m = 3$. We simulate using different values of $\mathbf{B}_1$ ($\bm{\Phi}_1$) and for $\mathbf{B}_2 = - \mathbf{B}_1$ ($\bm{\Phi}_2 = -\bm{\Phi}_1$ in the VTAR).}  
	\end{tabular}
\end{table}

\begin{table}[H]
	\centering
	\caption{\textbf{Power of the additive nonlinearity test at different sample sizes based on 1000 replications from a model with $m = 3$ regimes and $n = 5$, testing the null $\text{H}_0 \colon m = 2$}}\label{table:Powern5}
	\begin{tabularx}{0.67\textwidth}{lYYY}
		\toprule
		& $LM_3$ & $LM_3^{\text{resc}}$  & Wilks\\
		\midrule
		&\multicolumn{3}{c}{Panel A: VLSTAR}\\
		$T = 400$ & \\[2pt]
		$\rho_i \sim \mathcal{U}(0.3, 0.5)$& 95.2 & 91.6 & 94.0\\[5pt]
		$T = 600$ & \\[2pt]
		$\rho_i \sim \mathcal{U}(0.3, 0.5)$& 92.4 & 90.8 & 92.2\\[5pt]
		$T = 1000$ & \\[2pt]
		$\rho_i \sim \mathcal{U}(0.3, 0.5)$& 99.2 & 98.6 & 99.2\\[2pt]
\hline
		&\multicolumn{3}{c}{Panel B: VTAR}\\
		$T = 400$ & \\[2pt]
		$\rho_i \sim \mathcal{U}(0.3, 0.5)$& 93.2 & 91.6 & 76.6\\[5pt]
		$T = 600$ & \\[2pt]
		$\rho_i \sim \mathcal{U}(0.3, 0.5)$& 98.9 & 98.3 & 99.8\\[5pt]
		$T = 1000$ & \\[2pt]
		$\rho_i \sim \mathcal{U}(0.3, 0.5)$& 99.9 & 100.0 & 100.0\\
		\bottomrule
		\multicolumn{4}{p{0.65\textwidth}}{Note: Data are simulated from a model with $m = 3$ regimes and we test the null of $\text{H}_0 \colon m = 2$ against $H_1 \colon m = 3$. We simulate using different values of $\mathbf{B}_1$ ($\bm{\Phi}_1$), for $\mathbf{B}_2 = -\mathbf{B}_1$ ($\bm{\Phi}_2 = -\bm{\Phi}_1$ in the VTAR) and $\mathbf{B}_3 = -0.7 \cdot \mathbf{I}_n$ ($\bm{\Phi}_3 = -0.7 \cdot \mathbf{I}_n$ in the VTAR), with $c_1 = 2$ and $c_2 = 4$ in both the models, and $\gamma_1 = \gamma_2 = 2$ in the VLSTAR model. The nominal significance level is set equal to 0.05.}  
	\end{tabularx}
\end{table}

\begin{table}[H]
	\centering
	\caption{\textbf{Selection frequencies. The frequency is in per cent based on 1000 replications from a model with $m = 2$ regimes and $n = 3$}}\label{table:SelFreqm2n5}
	\begin{tabular}{lrrr@{\hskip 0.2cm}rrr@{\hskip 0.2cm}rrr}
		\toprule
		& \multicolumn{3}{c}{$LM_3$} & \multicolumn{3}{c}{$LM_3^{\text{resc}}$} & \multicolumn{3}{c}{Wilks}\\
		\cmidrule(lr){2-4} \cmidrule(lr){5-7} \cmidrule(lr){8-10}\\
		& $\hat{m}=1$  & $\hat{m}=2$ & $\hat{m}\geq3$ & $\hat{m}=1$  & $\hat{m}=2$ & $\hat{m}\geq3$& $\hat{m}=1$  & $\hat{m}=2$ & $\hat{m}\geq3$\\
		\midrule
		&\multicolumn{9}{c}{Panel A: VLSTAR}\\
		\multicolumn{3}{l}{$T = 400, \rho_i \sim \mathcal{U}(0.3, 0.5)$} \\[2pt]
		$\alpha = 0.10$& 0.0 & \underline{92.8} & 7.2  & 0.0 & \underline{96.7} & 3.3 & 3.7  & \underline{89.3} & 7.0\\
		$\alpha = 0.05$& 0.0 & \underline{96.3} & 3.7  & 0.1 & \underline{98.4} & 1.6 & 5.7  & \underline{90.8} & 3.5\\
		$\alpha = 0.01$& 0.0 & \underline{98.6} & 1.4  & 0.7 & \underline{99.8} & 0.2 & 12.9 & \underline{85.9} & 1.2\\[5pt]
		
		$T = 600, \rho_i \sim \mathcal{U}(0.3, 0.5)$ \\[2pt]
		$\alpha = 0.10$& 0.0 & \underline{92.2} & 7.8  & 0.0 & \underline{95.0} & 5.0 & 0.6  & \underline{91.8} & 7.6\\
		$\alpha = 0.05$& 0.0 & \underline{95.0} & 5.0  & 0.0 & \underline{97.0} & 3.0 & 0.8  & \underline{94.2} & 5.0\\
		$\alpha = 0.01$& 0.0 & \underline{98.2} & 1.8  & 0.0 & \underline{98.6} & 1.4 & 2.2  & \underline{96.2} & 1.6\\[5pt]
		
		$T = 1000, \rho_i \sim \mathcal{U}(0.3, 0.5)$ \\[2pt]
		$\alpha = 0.10$& 0.0 & \underline{96.8} & 3.2 & 0.0 & \underline{97.6} & 2.4 & 0.0 & \underline{96.8} & 3.2\\
		$\alpha = 0.05$& 0.0 & \underline{98.6} & 1.4 & 0.0 & \underline{99.1} & 1.0 & 0.0 & \underline{98.8} & 1.2\\
		$\alpha = 0.01$& 0.0 & \underline{99.6} & 0.4 & 0.0 & \underline{99.8} & 0.2 & 0.2 & \underline{99.6} & 0.4\\[5pt]
		\hline
&\multicolumn{9}{c}{Panel B: VTAR}\\
\multicolumn{3}{l}{$T = 400, \rho_i \sim \mathcal{U}(0.3, 0.5)$} \\[2pt]
$\alpha = 0.10$& 0.0 & \underline{89.7} & 10.3 & 0.0 & \underline{94.1} & 5.9 & 0.0 & \underline{89.7} & 10.3\\
$\alpha = 0.05$& 0.0 & \underline{92.6} & 7.4  & 0.0 & \underline{96.1} & 3.9 & 0.0 & \underline{92.8} & 7.2\\
$\alpha = 0.01$& 0.0 & \underline{96.3} & 3.7  & 0.0 & \underline{98.8} & 1.2 & 0.0 & \underline{96.3} & 3.7\\[5pt]

$T = 600, \rho_i \sim \mathcal{U}(0.3, 0.5)$ \\[2pt]
$\alpha = 0.10$& 0.0 & \underline{76.7} & 23.3  & 0.0 & \underline{83.1} & 16.9 & 0.0  & \underline{77.3} & 22.7\\
$\alpha = 0.05$& 0.0 & \underline{83.2} & 16.8  & 0.0 & \underline{88.3} & 11.7 & 0.0  & \underline{83.3} & 16.7\\
$\alpha = 0.01$& 0.0 & \underline{92.2} & 7.8   & 0.0 & \underline{94.7} & 5.3  & 0.0  & \underline{92.3} & 7.7\\[5pt]

$T = 1000, \rho_i \sim \mathcal{U}(0.3, 0.5)$ \\[2pt]
$\alpha = 0.10$& 0.0 & \underline{84.4} &15.6 & 0.0 & \underline{88.0} &12.0 & 0.0 & \underline{84.9} & 15.1\\
$\alpha = 0.05$& 0.0 & \underline{90.5} & 9.5 & 0.0 & \underline{92.9} & 7.1 & 0.0 & \underline{90.6} & 9.4\\
$\alpha = 0.01$& 0.0 & \underline{96.6} & 3.4 & 0.0 & \underline{97.9} & 2.1 & 0.2 & \underline{96.5} & 3.5\\[5pt]
		\bottomrule  
\multicolumn{10}{p{\textwidth}}{Note: Data are simulated from a model with $m = 2$ regimes. We simulate using different values of $\mathbf{B}_1$ ($\bm{\Phi}_1$) and for $\mathbf{B}_2 = - \mathbf{B}_1$ ($\bm{\Phi}_2 = -\bm{\Phi}_1$ in the VTAR). Underlined values denote the identified number of regimes.}  
\end{tabular}
\end{table}

\section{Empirical applications}\label{sec:Empirical}

\subsection{Interest rate term structure}
We first apply the sequential test procedure to the U.S. monthly interest rates data already used in \cite{Tsay1998}. The dataset contains the 3-month treasury bill rates ($Y_{1,t}$) and 3-year treasury notes ($Y_{2,t}$) for the period from June 1953 to September 2022 ($T = 832$). These represent the short-term and intermediate-term series in the term structure of the interest rates. To obtain weakly stationary time series, the data have been considered as growth rates, \textit{i.e.}, $\mathbf{y}_t = \left(y_{1,t}, y_{2,t}\right)'$, with $y_{i,t} = \ln(Y_{i,t}) - \ln(Y_{i,t-1})$ for $i = 1,2$. The plots of the time series are shown in Fig. \ref{fig:Tsay}.

As a candidate transition variable, we select the maturity spread computed as $x_t = \ln(Y_{1,t}) - \ln(Y_{2,t})$ since, according to the inverted yield curve theory \citep{Harvey1988}, this would reflect the business cycle of the U.S. economy. Following \cite{Tsay1998}, to avoid random fluctuations in the interest rates term structure, we use the 3-month moving average of $x_t$ as a transition variable, therefore $s_t = (x_t + x_{t-1} + x_{t-2})/3$ (see the green line in Fig. \ref{fig:Tsay}). To select the VLSTAR lag length, we use AIC and BIC criteria which suggest a single lag specification, \textit{i.e.}, $p = 1$. To understand if the results of our methodology are similar to the ones obtained with an alternative method, we compare them with the equation-by-equation approach proposed by \cite{camacho04}.

The results of the tests are reported in Table \ref{table:Interestrates}. It can be noticed in the first column that all the tests strongly reject the null hypothesis of linearity at any significance level, except for the test implemented in \cite{camacho04} which rejects at 5 and 10\%. This means that, in line with what was supposed and proved by \cite{Tsay1998}, a nonlinearity is present in the dynamics of the interest rates and the top-down procedure for the selection of the number of regimes presented in this paper can be applied. The procedure foresees testing the null hypothesis of $m = 2$ regimes against the alternative of $m = 3$. Once again, the null hypothesis is rejected in all the test statistics (second column of the table) introduced in this article, while the null cannot be rejected with the equation-by-equation test of \cite{camacho04}. According to the results from the system-based test statistics, a 2-regime model is not enough to consider all the nonlinearity in the model. The third column of the table reports the test statistics for the null of $m = 3$ regimes. None of the tests rejects the null hypothesis. Consequently, it can be deduced that the optimal number of regimes for these time series is three, which is also what was originally supposed in \cite{Tsay1998}. 

\begin{table}[H]
	\centering
	\caption{\textbf{Tests for linearity and additive nonlinearity in the empirical application with interest rates}}\label{table:Interestrates}
\begin{tabularx}{0.75\textwidth}{lYYY}
		\toprule
		 &  $H_0\colon m = 1$ &  $H_0\colon m = 2$ & $H_0\colon m = 3$\\[2pt]
		 &   Test statistic ($p$-value) & Test statistic ($p$-value) & Test statistic ($p$-value)\\
		\midrule
$LM_3$             &$\underset{(1.11e-15)}{98.399}$ & $\underset{(2.08e-18)}{107.819}$& $\underset{(0.770)}{8.185}$\\
$LM_3^{\text{resc}}$&$\underset{(4.89e-15)}{8.141}$  & $\underset{(1.11e-16)}{8.920}$  & $\underset{(0.775)}{0.677}$\\
Wilks              &$\underset{(3.33e-16)}{101.087}$& $\underset{( 3.54e-09)}{64.381}$  & $\underset{(0.875)}{6.766}$\\
Camacho            & $\underset{(0.011)}{25.805}$                              &    $\underset{(0.966)}{0.265}$                   & - \\
		\bottomrule  
		\multicolumn{4}{p{0.75\textwidth}}{Note: In the first column the linearity test is reported, while columns two and three belong to the sequential procedure for the identification of the number of regimes. $p$-values are reported between parentheses.}  
	\end{tabularx}
\end{table}

\begin{figure}[H]
	\centering
	\caption{Time series plots of U.S. interest rates.}\label{fig:Tsay}
	\includegraphics[width=0.8\linewidth]{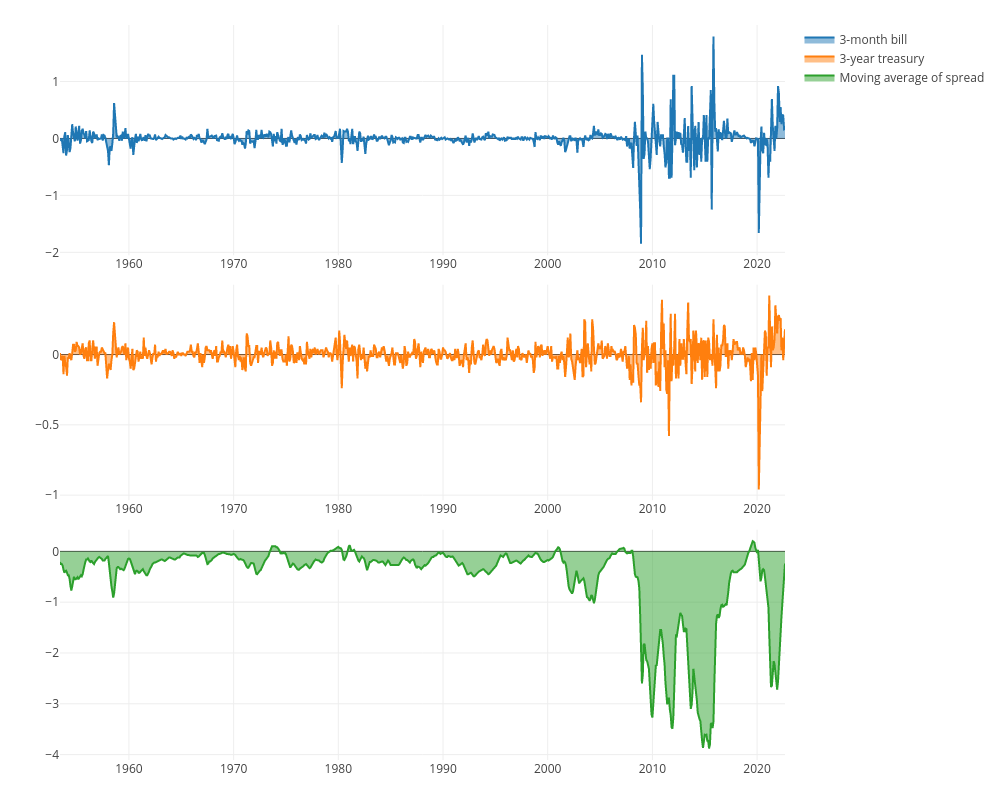} 
	\caption*{Note: The time series of the monthly 3-bill rate growth (blue line) and the 3-year Treasury rate growth (orange line) highlight the possible existence of several regimes. The 3-month moving average of the maturity spread (green line) is used as a transition variable.} 
\end{figure}

\subsection{River flows data}

As a second empirical example, the linearity and no remaining nonlinearity tests have been applied to the daily Icelandic river flow data for the period from 1972 to 1974. The time series in this dataset include river flows in cubic meters per second for two rivers, the J\"{o}kuls\'{a} and the Vatndals\'{a}, as well as the temperature and the precipitation, see Figure \ref{fig:Rivers}. River flow data has been shown to be nonlinear in several former applications. For instance, \cite{Tong1985} use a univariate threshold model to estimate their relationship with temperature and precipitation, while \cite{Tsay1998}, \cite{teya14}, and \cite{LivingstonJr2020} apply a multivariate nonlinear model. 

\begin{figure}[H]
	\centering
	\caption{Time series plots of Icelandic river flow data.}\label{fig:Rivers}
	\includegraphics[width=0.8\linewidth]{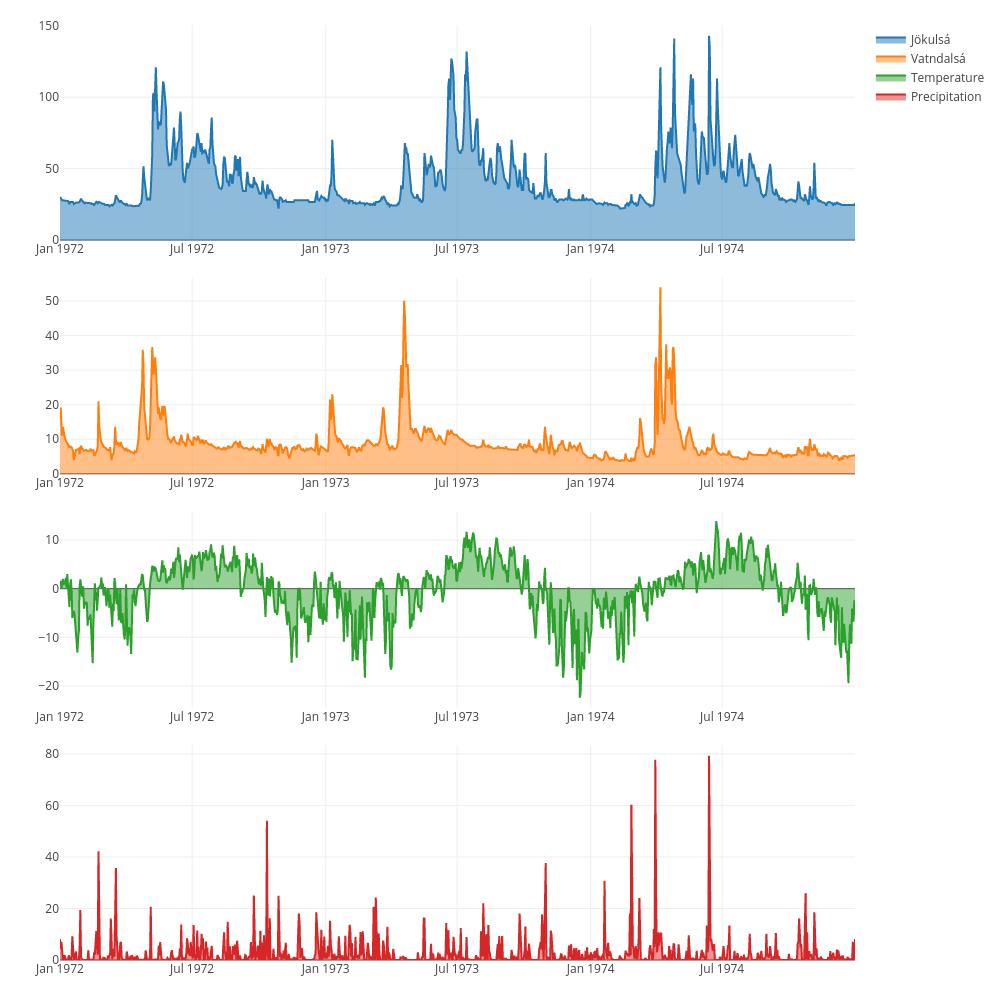} 
	\caption*{Note: The time series of the J\"{o}kuls\'{a} (blue line) and the Vatndals\'{a} (orange line) river flows highlight a nonlinear path for both rivers. The temperature (green line) and the precipitation level (red line) can be used as a transition variable.} 
\end{figure}

Following \cite{Tsay1998} and \cite{teya14}, we first select the lagged temperature as the transition variable for both flow equations (Panel A of Table \ref{table:Riverflows}) and we compute both the linearity test, and the sequential procedure to identify the number of regimes in these series. According to the linearity test statistics, the null hypothesis of linearity is strongly rejected. The procedure introduced in this paper foresees to use the additive nonlinearity test on an increasing number of regimes when the null of linearity is rejected. Therefore, we perform the additive nonlinearity test when the null hypothesis is $\text{H}_0 \colon  m = 2$. While the test by \cite{camacho04} does not reject the null hypothesis, the null of $m = 2$ is rejected at some significance level (\textit{i.e.}, $\alpha = 0.10, 0.05, 0.01$) in all the other tests, meaning that a residual nonlinearity is still present. The procedure iterates until a non-rejection is obtained at all the significance levels. We iterate again the procedure by testing the null hypothesis of $\text{H}_0 \colon m = 3$ regimes. In this case, the null hypothesis cannot be rejected for the standard and rescaled LM-type tests at any standard significance level, while Wilks' $\Lambda$ still rejects the null hypothesis if $\alpha = 0.10$. Based on these results, our top-down sequential procedure points to $m=3$ regimes as the optimal number of regimes for the Icelandic rivers' time series, which is in line with what found in \cite{teya14}.

We also apply the tests using the lagged precipitation as a transition variable (Panel B of Table \ref{table:Riverflows}). As for the case of the lagged temperature as a transition variable, the null hypothesis of linearity is rejected for all the tests. Once rejected the hypothesis of a linear model, we conduct the top-down procedure for the selection of the number of regimes. In this case, the additive nonlinearity tests are not able to reject the null hypothesis of $m = 2$ regimes. This means that, with the lagged precipitation as a transition variable, the optimal model is a 2-regime model.

\begin{table}[H]
	\centering
	\caption{\textbf{Tests for linearity and additive nonlinearity in the empirical application with river flows}}\label{table:Riverflows}
	\begin{tabularx}{0.75\textwidth}{lYYY}
		\toprule
		& \multicolumn{3}{c}{Panel A: Lagged temperature as the transition variable}\\
		\cmidrule{2-4}
		&  $H_0\colon m = 1$ &  $H_0\colon m = 2$ & $H_0\colon m = 3$\\[2pt]
		&   Test statistic ($p$-value) & Test statistic ($p$-value) & Test statistic ($p$-value)\\
		\midrule
		$LM_3$             &$\underset{(2.22e-15)}{89.301}$ & $\underset{(0.021)}{23.939}$& $\underset{(0.143)}{17.174}$\\
		$LM_3^{\text{resc}}$&$\underset{(2.31e-15)}{21.098}$  & $\underset{(0.022)}{1.984}$  & $\underset{(0.148)}{1.423}$\\
		Wilks              &$\underset{(4.44e-16)}{100.663}$& $\underset{(0.004)}{29.333}$  & $\underset{(0.065)}{20.095}$\\
		Camacho            &$\underset{(1.88e-10)}{71.231}$&$\underset{(0.289)}{14.184}$ & - \\

				\toprule
		& \multicolumn{3}{c}{Panel B: Lagged precipitation as the transition variable}\\
		\cmidrule{2-4}
		&  $H_0\colon m = 1$ &  $H_0\colon m = 2$ & $H_0\colon m = 3$\\[2pt]
		&   Test statistic ($p$-value) & Test statistic ($p$-value) & Test statistic ($p$-value)\\
		\midrule
		$LM_3$             &$\underset{(4.14e-14)}{90.393}$ & $\underset{(0.719)}{8.810}$& -\\
		$LM_3^{\text{resc}}$&$\underset{(1.06e-13)}{7.491}$ & $\underset{(0.722)}{0.730}$& -\\
		Wilks              & $\underset{(2.21e-14)}{91.807}$ & $\underset{(0.670)}{9.377}$& -\\
		Camacho            & $\underset{(0.007)}{27.222}$& $\underset{(0.478)}{11.606}$& -\\
		\bottomrule  
		\multicolumn{4}{p{0.75\textwidth}}{Note: In the first column the linearity test is reported, while columns two and three belong to the sequential procedure for the identification of the number of regimes. $p$-values are reported between parentheses.}  
	\end{tabularx}
\end{table}

\section{Conclusions}\label{sec:Conclusions}

In this paper, we developed a simple method for selecting the number of regimes in multivariate nonlinear models with no restrictions on the number of dependent variables and transitions, also giving a more formal context for the linearity and no additional nonlinearity tests introduced in \cite{yate14}.


The results on small-sample properties of the tests are of interest because they highlight that the empirical sizes are affected by the dimension of the model, the size of the sample and the persistence of the time series. We find that the standard LM tests tend to be size-distorted when the time series are almost non-stationary. We also show that Wilks' $\Lambda$ statistic has satisfying size properties, and is recommended for empirical use. Nevertheless, the size of the LM test can be adjusted using a proper bootstrapped version, although this has not been addressed in this work. Not surprisingly, the power experiments demonstrate that the joint test is more powerful in finite samples when the number of temporal observations is large. The selection frequencies of the number of regimes reflect what is already observed through empirical sizes and powers. The sequential procedure is capable of correctly identifying the number of regimes for any sample size, and for both smooth and abrupt regime changes. Finally, our simulation study underlines that the sequential procedure introduced in this paper can be applied either to detect the number of regimes in smoothly changing time series and abrupt regime-changing time series. 

When we apply the sequential test procedure to real data, we can observe that the tests introduced in this paper lead to more rejections with respect to the test introduced by \cite{camacho04}. On the one hand, our approach foresees that the VLSTAR model is correctly specified, and that deviations from linearity are due to remaining nonlinearity. In contrast, the test proposed by \cite{camacho04} may be less sensitive to misspecification, since each equation is tested separately and may have different sources of nonlinearity. On the other hand, more rejections may indicate that the tests on the whole system have greater power to detect nonlinearity, since they account for the joint behaviour of all the equations in the system. As a result, the system-based test may be more likely to identify nonlinear relationships that are present across multiple equations in the system.

A possible implementation for future research could foresee the use of the procedure to detect the number of structural breaks in a multivariate linear model. In fact, if the transition variable is a temporal trend, the VLSTAR model becomes a time-varying parameter model and the changes in regimes coincide with smooth structural breaks.
 
\section*{Acknowledgements}
I offer my sincerest gratitude to Timo T\"{e}rasvirta who has contributed to the early version of this paper, and whose comments have helped me improve the theoretical background behind the sequential procedure. I retain the responsibility for any errors and shortcomings in this work.

\newpage

\bibliographystyle{apalike}
\bibliography{Sequential}

\pagebreak

\appendix

\section{The score vector}\label{sec:Score}

Let's assume a 3-regime model with a single lag
\begin{equation}\label{eq:VLSTAR3regimesa}
	\mathbf{y}_t = \bm{\mu}_0 + \mathbf{\Phi}_0 \mathbf{y}_{t-1} + \mathbf{G}_t^{(1)}\left(\bm{\mu}_1 + \mathbf{\Phi}_1\mathbf{y}_{t-1}\right)+ \mathbf{G}_t^{(2)}\left(\bm{\mu}_2 + \mathbf{\Phi}_2\mathbf{y}_{t-1}\right) + \bm{\varepsilon}_t
\end{equation}
which can be written as
\begin{equation}\label{eq:VLSTAR3regimesb}
	\mathbf{y}_t = \mathbf{\Psi}_t'\mathbf{B}' \mathbf{x}_t + \mathbf{G}_t^{(2)}\mathbf{B}_2\mathbf{x}_t + \bm{\varepsilon}_t
\end{equation}
where $\mathbf{\Psi} = \left(\mathbf{I}_n, \mathbf{G}_t^{(1)}\right)'$ is a $2n \times n$ matrix, $\mathbf{B} = \left(\mathbf{B}_0, \mathbf{B}_1\right)$ is a $(1+n) \times 2n$ matrix of parameters, $\mathbf{x}_t = \left[1, \mathbf{y}_{t-1}'\right]'$ is a $(n+1) \times 1$ vector, and $\mathbf{B}_j = \left(\bm{\mu}_j', \mathbf{\Phi}_j'\right)'$, for $j = 0, 1, 2$. 

To derive the score and the relevant derivatives, we first assume a single transition model, \textit{i.e.} excluding $\mathbf{G}_t^{(2)}\left(\bm{\mu}_2 + \mathbf{\Phi}_2\mathbf{y}_{t-1}\right)$ from model \eqref{eq:VLSTAR3regimesa}. The score of the log-likelihood for model \eqref{eq:VLSTAR3regimesa} is
\begin{equation*}
	\mathbf{s}_t(\bm{\theta}) = \frac{\partial \ell_t(\bm{\theta})}{\partial \bm{\theta}} = -\frac{\partial \bm{\varepsilon}_t'}{\partial \bm{\theta}}\mathbf{\Omega}^{-1}\bm{\varepsilon}_t.
\end{equation*}
Now, we have to compute $\partial \bm{\varepsilon}_t'/\partial \bm{\theta}$, where
\begin{equation*}
	\bm{\varepsilon}_t = \mathbf{y}_t - \bm{\mu}_0 - \mathbf{\Phi}_0 \mathbf{y}_{t-1} - \mathbf{G}_t(\bm{\mu}_1 + \mathbf{\Phi}_1 \mathbf{y}_{t-1}).
\end{equation*}
The $j$-th equation (out of $n$) is
\begin{align*}
	\bm{e}_j\bm{\varepsilon}_t &= \bm{e}_j'\left(\mathbf{y}_t - \bm{\mu}_0 - \mathbf{\Phi}_0\mathbf{y}_{t-1} - \mathbf{G}_t(\bm{\mu}_1 + \mathbf{\Phi}_1 \mathbf{y}_{t-1})\right)\\
	&= \bm{e}_j'\left(\mathbf{y}_t - \bm{\mu}_0 - \sum_{i=1}^{n}\phi_{0i}y_{i,t-1} - \mathbf{G}_t(\bm{\mu}_1 + \sum_{i=1}^{n}\phi_{1i}y_{i,t-1})\right)
\end{align*}
where $\mathbf{e}_j$ is the $j$-th column of the $n \times n$ identity matrix. Denote the parameters of the $i$-th equation as $\bm{\theta}_i$. Then, $\partial \mathbf{e}_j'\bm{\varepsilon}_t/\partial \bm{\theta}_i = \mathbf{0}$ is the $2n(n+1)$-dimensional null vector. For $i = j$, one obtains, component by component,
\begin{equation*}
	\frac{\partial \mathbf{e}_j'\bm{\varepsilon}_t}{\partial \mu_{0j}} = 1
\end{equation*}
so $\partial \bm{\varepsilon}_t/\partial \mu_{0j} = \mathbf{e}_j$ and $\partial \bm{\varepsilon}_t'/\partial \bm{\mu}_0 = \mathbf{I}_n$. Similarly, it holds that $\partial \bm{\varepsilon}_t'/\partial \bm{\mu}_1 = \mathbf{G}_t$. Further, $\partial \bm{\varepsilon}_t/\partial \phi_{0ij} = -\partial\mathbf{\Phi}_0\mathbf{y}_{t-1}/\partial \phi_{0ij} = -\mathbf{e}_i\mathbf{e}_j'\mathbf{y}_{t-1}$, and with a slight abuse of notation it follows that
\begin{align}
	\frac{\partial \bm{\varepsilon}_t'}{\partial \mathbf{\Phi}_0} &= - \begin{bmatrix}
		\mathbf{e}_1\mathbf{e}_1' & \dots & \mathbf{e}_1\mathbf{e}_n'\\
		\vdots & \vdots & \vdots\\
		\mathbf{e}_n\mathbf{e}_1' & \dots & \mathbf{e}_n\mathbf{e}_n'
	\end{bmatrix} \otimes \mathbf{y}_{t-1} \nonumber\\
	&= - \begin{bmatrix}
		y_{1,t-1}\mathbf{e}_1\mathbf{e}_1' & \dots & y_{1,t-1}\mathbf{e}_1\mathbf{e}_n'\\
		\vdots & \vdots & \vdots\\
		y_{1,t-1}\mathbf{e}_n\mathbf{e}_1' & \dots & y_{1,t-1}\mathbf{e}_n\mathbf{e}_n'\\	
		\vdots & \vdots & \vdots\\
		y_{n,t-1}\mathbf{e}_1\mathbf{e}_1' & \dots & y_{1,t-1}\mathbf{e}_1\mathbf{e}_n'\\
		\vdots & \vdots & \vdots\\
		y_{n,t-1}\mathbf{e}_n\mathbf{e}_1' & \dots & y_{n,t-1}\mathbf{e}_n\mathbf{e}_n' \label{eq:derivphi0}		
	\end{bmatrix}
\end{align}
and
\begin{align}
	\frac{\partial \bm{\varepsilon}_t'}{\partial \mathbf{\Phi}_1} &= - \begin{bmatrix}
		\mathbf{G}_t\mathbf{e}_1\mathbf{e}_1' & \dots & \mathbf{G}_t\mathbf{e}_1\mathbf{e}_n'\\
		\vdots & \vdots & \vdots\\
		\mathbf{G}_t\mathbf{e}_n\mathbf{e}_1' & \dots & \mathbf{G}_t\mathbf{e}_n\mathbf{e}_n'
	\end{bmatrix} \otimes \mathbf{y}_{t-1} \nonumber\\
	&= - \begin{bmatrix}
		g_{1t}y_{1,t-1}\mathbf{e}_1\mathbf{e}_1' & \dots & g_{1t}y_{1,t-1}\mathbf{e}_1\mathbf{e}_n'\\
		\vdots & \vdots & \vdots\\
		g_{nt}y_{1,t-1}\mathbf{e}_n\mathbf{e}_1' & \dots & g_{nt}y_{1,t-1}\mathbf{e}_n\mathbf{e}_n'\\	
		\vdots & \vdots & \vdots\\
		g_{1t}y_{n,t-1}\mathbf{e}_1\mathbf{e}_1' & \dots & g_{1t}y_{1,t-1}\mathbf{e}_1\mathbf{e}_n'\\
		\vdots & \vdots & \vdots\\
		g_{nt}y_{n,t-1}\mathbf{e}_n\mathbf{e}_1' & \dots & g_{nt}y_{n,t-1}\mathbf{e}_n\mathbf{e}_n'				
	\end{bmatrix}. \label{eq:derivphi1}	
\end{align}
The matrices $\mathbf{e}_1\mathbf{e}_1'$ have only one nonzero element that equals one, so for the score, the corresponding $n^2 \times n$ matrix becomes $\mathbf{y}_{t-1} \otimes \mathbf{I}_n$. Terms for the average score from \eqref{eq:derivphi0} are as a vector
\begin{equation*}
	\mathbf{s}_T(\mathbf{\Phi}_0) = \frac{1}{T}\sum_{t=1}^{T}(\mathbf{y}_{t-1} \otimes \mathbf{I}_n)
\end{equation*}
which is an $n^2 \times 2$ matrix, and from \eqref{eq:derivphi1} it follows that
\begin{equation*}
	\mathbf{s}_T(\mathbf{\Phi}_1) = \frac{1}{T}\sum_{t=1}^{T}(\mathbf{G}_t\mathbf{y}_{t-1} \otimes \mathbf{I}_n).
\end{equation*}
Finally,
\begin{align*}
	\frac{\partial \mathbf{G}_t}{\partial \gamma_i} &= \frac{\partial \ \text{diag}(g_{1t}, \ldots, g_{nt})}{\gamma_i} = \text{diag}\left(0, \ldots, \frac{\partial g_{it}}{\partial \gamma_i}, \ldots, 0\right)\\
	&= \frac{\partial g_{it}}{\partial \gamma_i}\text{diag}(0, \ldots, 1, \ldots) = \frac{\partial g_{it}}{\gamma_i} \mathbf{e}_i\mathbf{e}_i'
\end{align*}
where $\partial g_{it}/\partial \gamma_i = g_{it}(1-g_{it})(s_{i,t}-c_{i})$. This implies
\begin{equation*}
	\frac{\partial \bm{\varepsilon}_t}{\partial \gamma_i} = -\frac{\partial g_{it}}{\partial \gamma_i}\mathbf{e}_i\mathbf{e}_i'(\bm{\mu}_1 + \mathbf{\Phi}_1\mathbf{y}_{t-1}).
\end{equation*}
For the $i$-th error, the derivative w.r.t. $\gamma_i$ equals
\begin{equation*}
	\frac{\partial \varepsilon_{it}}{\partial \gamma_i} =  -\frac{\partial g_{it}}{\partial \gamma_i}(\bm{\mu}_1 + \mathbf{\Phi}_1\mathbf{y}_{t-1}).
\end{equation*}
Denoting $\mathbf{g}_{\bm{\gamma},t} = \left(\partial g_{1t}/\partial \gamma_1, \ldots, \partial g_{nt}/\partial \gamma_n\right)'$ and $\bm{\gamma} = \left(\gamma_1, \ldots, \gamma_n\right)'$, one can write, after removing the zero elements and using the notation $\partial \bm{\varepsilon}_t/\partial \bm{\gamma} = \left(\partial \varepsilon_{1t}/\partial \gamma_1, \ldots, \partial \varepsilon_{nt}/\partial \gamma_n\right)'$,
\begin{equation*}
	\frac{\partial \bm{\varepsilon}_t'}{\partial \bm{\gamma}} = -\mathbf{g}_{\bm{\gamma},t}' \otimes (\bm{\mu}_1 + \mathbf{\Phi}_1 \mathbf{y}_{t-1}),
\end{equation*}
which is an $n \times n$ matrix. Now, $\partial g_{1t}/\partial c_i = -\gamma_i g_{it}(1-g_{it})$. Analogously to the previous notation, setting $\mathbf{c} = \left(c_1, \ldots, c_n\right)'$ and $\mathbf{g}_{\mathbf{c},t} = \left(\partial g_{1t}/\partial c_1, \ldots, \partial g_{nt}/\partial c_n\right)'$, one obtains
\begin{equation*}
	\frac{\partial \bm{\varepsilon}_t'}{\partial \mathbf{c}} = -\mathbf{g}_{\mathbf{c},t} \otimes (\bm{\mu}_1 + \mathbf{\Phi}_1 \mathbf{y}_{t-1}).
\end{equation*}
Drawing things together, the score vector equals (assuming $\mathbf{y}_0$ known)
\begin{equation}\label{eq:averagescore}
	\begin{bmatrix}
		\mathbf{s}_T(\mathbf{\Phi}_0)\\ \mathbf{s}_T(\mathbf{\Phi}_1)\\ \mathbf{s}_T(\bm{\gamma}) \\ \mathbf{s}_T(\mathbf{c})
	\end{bmatrix} = -\frac{1}{T}\sum_{t=1}^{T}\begin{bmatrix}
		\left(\mathbf{I}_n \otimes \mathbf{y}_{t-1}\right)\mathbf{\Omega}_t^{-1}\bm{\varepsilon}_t\\
		\left(\mathbf{I}_n \otimes \mathbf{G}_t\mathbf{y}_{t-1}\right)\mathbf{\Omega}_t^{-1}\bm{\varepsilon}_t\\
		\left\{\mathbf{g}_{\bm{\gamma},t}' \otimes (\bm{\mu}_1 + \mathbf{\Phi}_1\mathbf{y}_{t-1})\right\}\mathbf{\Omega}^{-1}_t\bm{\varepsilon}_t\\
		\left\{\mathbf{g}_{\mathbf{c},t}' \otimes (\bm{\mu}_1 + \mathbf{\Phi}_1\mathbf{y}_{t-1})\right\}\mathbf{\Omega}^{-1}_t\bm{\varepsilon}_t
	\end{bmatrix}.
\end{equation}
When the model is the one in Eq. \eqref{eq:VLSTAR3regimesa}, the purpose could be testing one transition against two. In this case, the term $\mathbf{G}_t^{(2)}\left(\bm{\mu}_2 + \mathbf{\Phi}_2\mathbf{y}_{t-1}\right)$ can be approximated through a third-order Taylor expansion, so that the error vector becomes
\begin{equation*}
	\bm{\varepsilon}_t = \mathbf{y}_t - \bm{\mu}_0 - \mathbf{\Phi}_0 \mathbf{y}_{t-1} - \mathbf{G}_t^{(1)} - \sum_{l = 1}^{L=3}\mathbf{\Upsilon}_l(\mathbf{y}_{t-1} \odot \mathbf{s}_t^{l})
\end{equation*}
where $\mathbf{\Upsilon}_l = \text{diag}\left(\upsilon_{l1}, \ldots, \upsilon_{ln}\right)$ and $\mathbf{s}_t^{l} = \left(s_{1t}^l, \ldots, s_{nt}^l\right)'$, $l = 1, 2, 3$. Then, $\partial \mathbf{\Upsilon}_l/\partial \upsilon_{lj} = \mathbf{e}_j\mathbf{e}_j'$ and $\partial \bm{\varepsilon}_t/\partial \upsilon_{lj} = \mathbf{e}_j'(\mathbf{y}_{t-1} \odot \mathbf{s}_t^{i})$. Finally, defining
\begin{equation*}
	\frac{\partial \bm{\varepsilon}_t}{\partial \bm{\upsilon}_l} = \left(\frac{\partial \bm{\varepsilon}_t}{\partial \upsilon_{l1}}, \ldots, \frac{\partial \bm{\varepsilon}_t}{\partial \upsilon_{ln}}\right)', \quad \text{for } l= 1,2,3
\end{equation*}
the corresponding components for the average score are
\begin{equation*}
	\mathbf{s}_T(\bm{\upsilon}_l) = \frac{1}{T}\sum_{t=1}^{T}\left\{\left(\mathbf{y}_{t-1} \odot \mathbf{s}_t^{l}\right) \otimes \mathbf{I}_n\right\} \mathbf{\Omega}_t^{-1}\bm{\varepsilon}_t.
\end{equation*}
Equivalently, the average score for the special case of testing linearity against a single transition can be obtained by suppressing the redundant rows from \eqref{eq:averagescore}.

\section{Proof of Theorem 1}
\label{sec:AppendixA}

\begin{proof}
Assuming a first-order Taylor expansion, the Lagrange multiplier test under the null is derived from the score matrix
\begin{equation}\label{eq:Scoreapp}
	\frac{\partial \ell_T(\hat{\bm{\theta}})}{\partial \mathbf{D}_1} = \sum_{t=1}^{T}\left\{\mathbf{x}_t\mathbf{s}_t\left(\mathbf{y}_t - \mathbf{\hat{D}}_0'\mathbf{x}_t\right)' \hat{\mathbf{\Omega}}^{-1}\right\}.
\end{equation}
Setting $\mathbf{Y} = \left[\mathbf{y}_1' \ \mathbf{y}_2' \ \ldots \ \mathbf{y}_T'\right]'$, $\mathbf{Z} = \left[x_1's_1 \ x_2's_2 \ \ldots  \ x_T's_T\right]$, we have that
\begin{equation}\label{eq:Scoreapp2}
	\frac{\partial \ell_T(\hat{\bm{\theta}})}{\partial \mathbf{D}_1} = \mathbf{Z}' \left(\mathbf{Y} - \mathbf{X}\mathbf{\hat{D}}_0\right)\hat{\mathbf{\Omega}}^{-1},
\end{equation}
where $\hat{\mathbf{D}}_0$ and $\hat{\mathbf{\Omega}}$ are estimates under the null hypothesis. Under regularity conditions, the score converges in probability to a matrix-normal distribution with zero mean and variance $\mathbf{Z}'(\mathbf{I} - \mathbf{P})Z \otimes \mathbf{\Omega}^{-1}$ conditional on $\mathbf{X}$ and $\mathbf{Z}$, where $\mathbf{P}_X \equiv \mathbf{X}(\mathbf{X}'\mathbf{X})^{-1}\mathbf{X}$ is the projection matrix of $\mathbf{X}$.

To see this, we write \eqref{eq:Scoreapp2} as follows
\begin{align*}
\mathbf{Q} &\equiv \frac{\partial \ell_T(\hat{\bm{\theta}})}{\partial \mathbf{D}_1} = \mathbf{Z}'\left(\mathbf{Y} - \mathbf{X}\hat{\mathbf{D}}_0\right)\hat{\mathbf{\Omega}}^{-1}\\
&= \mathbf{Z}'\left(\mathbf{Y} - \mathbf{X}(\mathbf{X}'\mathbf{X})\mathbf{X}'\mathbf{Y}\right)\hat{\mathbf{\Omega}}^{-1}\\
&= \mathbf{Z}'\left(\mathbf{I} - \mathbf{P}_X\right)\left(\mathbf{X}\mathbf{D}_0 + \mathbf{E}\right)\hat{\mathbf{\Omega}}^{-1}\\
&= \mathbf{Z}'\left(\mathbf{I} - \mathbf{P}_X\right)\mathbf{E}\hat{\mathbf{\Omega}}^{-1}.
\end{align*}
Under the null hypothesis, $\mathbf{Y} = \mathbf{X}\mathbf{D}_0 + \mathbf{E}$, where $\mathbf{E} = \left[\bm{\varepsilon}_1' \ \ldots \ \bm{\varepsilon}_T'\right]'$ and $\text{vec}(\mathbf{E}')$ follows a $\mathcal{N}\left(\mathbf{0}, \mathbf{I}_T\otimes \mathbf{\Omega}\right)$ distribution. Under the null hypothesis, $\hat{\mathbf{\Omega}}$ will converge in probability to $\mathbf{\Omega}$. Set
\begin{equation*}
\mathbf{S} = \left(\mathbf{Z}'(\mathbf{I} - \mathbf{P}_X)\mathbf{Z}\right)^{-1/2}\mathbf{Q}\hat{\mathbf{\Omega}}^{1/2}
\end{equation*}
which will asymptotically converge to a zero-mean matrix-normal distribution with variance $\mathbf{I} \otimes \mathbf{I}$. Thus, we have the chi-square version LM test statistic
\begin{equation*}
LM  = \text{tr}\left\{\mathbf{S}'\mathbf{S}\right\} = \text{tr}\left\{\hat{\mathbf{\Omega}}^{-1}(\mathbf{Y} - \mathbf{X}\hat{\mathbf{D}})'\mathbf{Z}[\mathbf{Z}'(\mathbf{I}_T - \mathbf{P}_X)\mathbf{Z}]^{-1}\mathbf{Z}'(\mathbf{Y}-\mathbf{X}\hat{\mathbf{D}}_0)\right\}
\end{equation*}
which converges to the $\chi^2_{n(1+np)}$-distribution when the null hypothesis is valid.
\end{proof}

\section{First-order partial derivatives of $\mathbf{\Psi}_t' \mathbf{B}' \mathbf{x}_t$}
\label{sec:firstderiv}

The vectorised first-order derivative of $\mathbf{\Psi}_t' \mathbf{B}' \mathbf{x}_t$ with respect to parameters $\bm{\theta}$ can be easily found in both univariate and multivariate cases, see \cite{Eitrheim1996} and \cite{yate14}. The set of parameters $\bm{\theta}$ consists of $\text{vec}(\mathbf{B})$, $\bm{\gamma}$ and $\mathbf{C}$, where $\mathbf{B} = \left[b_{ij}\right]$, $\bm{\gamma} = \left[\gamma_{ij}\right]$ and $\mathbf{C} = \left[c_{ij}\right]$.

For the $ij$-th parameter of $\mathbf{B}$, $b_{ij}$, we have
\begin{equation}\label{eq:firstder}
	\frac{\partial \mathbf{\Psi}_t' \mathbf{B}' \mathbf{x}_t}{\partial b_{ij}} = \mathbf{\Psi}_t'\mathbf{H}_{ij}\mathbf{x}_t
\end{equation}
where $\mathbf{H}_{ij} = \left[h_{kl}\right]$ is a matrix in which $h_{ij} = 1$ and $h_{kl} = 0$ for $k \neq i$ and $l \neq i$. The first derivative in \eqref{eq:firstder} is the directional derivative of $\mathbf{\Psi}_t' \mathbf{B}' \mathbf{x}_t$ with respect to the unit length matrix $\mathbf{H}_{ij}$.

For the parameter matrices $\bm{\gamma} = \left[\gamma_{ij}\right]$ and $\mathbf{C} = \left[c_{ij}\right]$, letting $\beta_{ij} = \gamma_{ij}, c_{ij}$, we have
	\begin{equation}
	\frac{\partial \mathbf{\Psi}_t' \mathbf{B}' \mathbf{x}_t}{\partial \beta_{ij}} = \left(0_n, \ldots, \frac{\partial \mathbf{G}_t^{(d)}}{\partial \beta_{ij}}, \ldots, 0_n\right)\mathbf{B}' \mathbf{x}_t  = \frac{\partial \mathbf{G}_t^{(d)}}{\partial \beta_{ij}}\mathbf{B}_{d+1}'\mathbf{x}_t
	\end{equation}
for $d = 1, \ldots, m-1$, where
\begin{equation}
\frac{\partial \mathbf{G}_t^{(d)}}{\beta_{ij}} = \text{diag}\left\{0, \ldots, \frac{\partial g_t^{ij}}{\beta_{ij}}, \ldots, 0\right\}
\end{equation}
for $j = 1, \ldots, n$. When $\beta_{ij} = \gamma_{ij}$,
\begin{equation}
\frac{\partial g^{ij}_t}{\gamma_{ij}} = \left(g_t^{ij}\right)^2 \exp \left\{-\gamma_{ij}\left(s_t - c_{ij}\right)\right\}\left(s_t - c_{ij}\right) = \left(s_t - c_{ij}\right)g_t^{ij}\left(1 - g_t^{ij}\right)
\end{equation}
and when $\beta_{ij} = c_{ij}$,
\begin{equation}
\frac{\partial g_t^{ij}}{\partial c_{ij}} = -\left(g_t^{ij}\right)^2 \text{exp}\left\{-\gamma_{ij}\left(s_t - c_{ij}\right)\right\}\gamma_{ij} = -\gamma_{ij} g_t^{ij}\left(1 - g_t^{ij}\right).
\end{equation}
The dimension of the first-order derivative of $\mathbf{\Psi}_t' \mathbf{B}' \mathbf{x}_t$ with respect to $\bm{\theta}$ is $n \times \left[(pn + 1)mn + 2(m-1)n\right]$.


\section{Wilks' extension of the test}
\label{sec:Ftests}
As in \cite{yate14}, we use an improvement to the rescaled LM-type test presented in Eq. \eqref{eq:LMadj} based on the so-called Wilks' $\Lambda$-distribution. 

\begin{theorem}
Let $\mathbf{\text{RSS}}_0$ and $\mathbf{\text{RSS}}_1$ be the $n \times n$ residual sum of squares from respectively the restricted and the auxiliary regression, and let $\mathbf{W}_1 = \mathbf{\text{RSS}}_0 - \mathbf{\text{RSS}}_1$ and $\mathbf{W}_2 = \mathbf{\text{RSS}}_1$. In the linearity test, under the null, $\mathbf{W}_1$ and $\mathbf{W}_2$ are two independent Wishart-distributed random matrices 
\begin{equation}
\mathbf{W}_1 \sim \mathcal{W}_n\left(\mathbf{\Omega}, \text{cd}(\mathbf{Z})\right), \qquad \mathbf{W}_2 \sim \mathcal{W}_n \left(\mathbf{\Omega}, T - \text{cd}(\mathbf{X}) - \text{cd}(\mathbf{Z})\right)
\end{equation}
where $\text{cd}(\cdot)$ is the column dimension of a matrix.
\end{theorem}

\begin{proof}
The score matrix evaluated under the null hypothesis has the general form
\begin{equation*}
\frac{\partial \ell_T(\hat{\bm{\theta}})}{\partial \mathbf{D}_1} = \mathbf{Z}'\left(\mathbf{Y} - \mathbf{X}\hat{\mathbf{D}}_0\right)\hat{\mathbf{\Omega}}^{-1}.
\end{equation*}
Using the auxiliary approach for computing the test statistic produces two residual sum of squares, $\mathbf{\text{RSS}}_0$ and $\mathbf{\text{RSS}}_1$. The former is the residual sum of the squares matrix from the restricted regression, \textit{i.e.}, $\mathbf{\text{RSS}}_0 = \hat{\mathbf{E}}'\hat{\mathbf{E}}$, $\hat{\mathbf{E}} = (\mathbf{I} - \mathbf{P}_X)\mathbf{Y}$, where $\mathbf{P}_X$ is the projection matrix of $\mathbf{X}$. Notice that under the null hypothesis, $\mathbf{Y} = \mathbf{X}\mathbf{D}_0 + \mathbf{E}$, where $\mathbf{E} = \left[\bm{\varepsilon}_1' \ \ldots \ \bm{\varepsilon}_T'\right]'$, and $\text{vec}(\mathbf{E}') \sim \mathcal{N}(\mathbf{0}, \mathbf{I}_T \otimes \mathbf{\Omega})$. Consequently, we have that $\hat{\mathbf{E}} = (\mathbf{I} - \mathbf{P}_X)\mathbf{Y} = (\mathbf{I} - \mathbf{P}_X)\mathbf{E}$. $\mathbf{RSS}_1$ is the residual sum of the squares matrix from the auxiliary regression, \textit{i.e.}, $\mathbf{RSS}_1 = \hat{\mathbf{\Xi}}'\hat{\mathbf{\Xi}}$, where $\hat{\mathbf{\Xi}} = (\mathbf{I} - \mathbf{P}_{XZ})\hat{\mathbf{E}}$, where $\mathbf{P}_{XZ}$ is the projection matrix of the matrix $\left[\mathbf{X}, \mathbf{Z}\right]$, such that
\begin{equation*}
\mathbf{P}_{XZ} = \left[\mathbf{X} \ \mathbf{Z}\right] = \begin{bmatrix}
\mathbf{X'X} & \mathbf{X'Z}\\
\mathbf{Z'X} & \mathbf{Z'Z}
\end{bmatrix}^{-1} \begin{bmatrix}
\mathbf{X'} \\ \mathbf{Z'}
\end{bmatrix}.
\end{equation*}
Therefore, $\mathbf{W}_1$ can be reformulated as follows
\begin{align*}
\mathbf{W}_1 &= \mathbf{\text{RSS}}_0 - \mathbf{\text{RSS}}_1 = \hat{\mathbf{E}}'\hat{\mathbf{E}} - \hat{\mathbf{\Xi}}'\hat{\mathbf{\Xi}}\\
&=  \hat{\mathbf{E}}'\mathbf{P}_{XZ} \hat{\mathbf{E}} =  \hat{\mathbf{E}}'\mathbf{Z}(\mathbf{Z}'(\mathbf{I}_T - \mathbf{P}_X)\mathbf{Z})^{-1}\mathbf{Z}'\hat{\mathbf{E}}\\
&= \mathbf{E}'(\mathbf{I}_T - \mathbf{P}_X)\mathbf{Z}(\mathbf{Z}'(\mathbf{I}_T - \mathbf{P}_X)\mathbf{Z})^{-1}\mathbf{Z}'(\mathbf{I}_T - \mathbf{P}_X)\mathbf{E}.
\end{align*}
Let $\mathbf{I}_T - \mathbf{P}_X = \mathbf{R}'\mathbf{R}$, where $\mathbf{R}$ is orthogonal to $\mathbf{X}$ and $\mathbf{R}'\mathbf{R} = \mathbf{I}_{T - \text{cd}(X)}$, then
\begin{equation*}
\mathbf{W}_1 = \mathbf{E}'\mathbf{R}'\mathbf{R}\mathbf{Z}(\mathbf{Z}'\mathbf{R}\mathbf{R}'\mathbf{Z})^{-1}\mathbf{Z}'\mathbf{R}\mathbf{R}'\mathbf{E}.
\end{equation*}
Setting $\mathbf{V}_1 = \mathbf{Z}'\mathbf{R}\mathbf{R}'\mathbf{E}$, we have that $\mathbf{V}_1 \sim \mathcal{N}(\mathbf{0}, \mathbf{Z'RR'Z} \otimes \mathbf{\Omega})$, therefore $\mathbf{W}_1$ follows a Wishart distribution generated by $\mathbf{V}_1$:
\begin{equation*}
	\mathbf{W}_1 = \mathbf{V}_1'(\mathbf{Z'RR'Z} )^{-1}\mathbf{V}_1 \sim \mathcal{W}_n(\mathbf{\Omega}, \text{cd}(\mathbf{Z})).
\end{equation*}
For $\mathbf{W}_2$, we obtain
\begin{align*}
\mathbf{W}_2 &= \mathbf{RSS}_1 = \hat{\mathbf{\Xi}}'\hat{\mathbf{\Xi}} = \hat{\mathbf{E}}'(\mathbf{I} - \mathbf{P}_{XZ})\hat{\mathbf{E}} = \hat{\mathbf{E}}'\hat{\mathbf{E}} - \hat{\mathbf{E}}'\mathbf{P}_{XZ}\hat{\mathbf{E}}\\
&= \hat{\mathbf{E}}'\hat{\mathbf{E}} - \hat{\mathbf{E}}'\mathbf{Z}(\mathbf{Z}'(\mathbf{I} - \mathbf{P}_X)\mathbf{Z})^{-1}\mathbf{Z}'\hat{\mathbf{E}}\\
&= \mathbf{E}'\mathbf{RR'(I-Z(Z'RR'Z)^{-1}}\mathbf{Z'})\mathbf{RR'E}\\
&= \mathbf{E'R}(\mathbf{I}-\mathbf{R'Z}(\mathbf{Z'RR'Z})^{-1}\mathbf{Z'R})\mathbf{R'E}.
\end{align*}
By imposing $\mathbf{I}_{T-\text{cd}(X)} - \mathbf{R'Z}(\mathbf{Z'RR'Z})^{-1}\mathbf{Z'R} = \mathbf{QQ'}$, where $\mathbf{Q} \bot \mathbf{R'Z}$ (with $\bot$ indicating orthogonality) and $\mathbf{QQ'} = \mathbf{I}_{T-\text{cd}(X)-\text{cd}(Z)}$, we have that $\mathbf{W}_2 = \mathbf{E'RQQ'R'E}$. Setting $\mathbf{V}_2 = \mathbf{Q'R'E}$, we have that $\mathbf{V}_2 \sim \mathcal{N}(\mathbf{0}, \mathbf{I}\otimes \mathbf{\Omega})$ and that $\mathbf{W}_2$ follows a Wishart distribution generated by $\mathbf{V}_2$, such that
\begin{equation*}
	\mathbf{W}_2 = \mathbf{V}_2'\mathbf{V}_2 \sim \mathcal{W}_n(\mathbf{\Omega}, T-\text{cd}(X)-\text{cd}(Z)).
\end{equation*}
Stacking the columns of $\mathbf{V}_1$ and $\mathbf{V}_2$ yields the random matrix
\begin{equation*}
\mathbf{U} = \begin{pmatrix}
	\mathbf{V}_1 \\ \mathbf{V}_2
\end{pmatrix} = \begin{pmatrix}
\mathbf{Z'R} \\ \mathbf{Q'}
\end{pmatrix}\mathbf{R'E}.
\end{equation*}
It follows that $\mathbf{U} \sim \mathcal{N}(\mathbf{0}, \mathbf{\Sigma} \otimes \mathbf{\Omega})$, where the row covariance matrix is
\begin{equation*}
\mathbf{\Sigma} = \begin{pmatrix}
	\mathbf{Z'R} \\ \mathbf{Q'}
\end{pmatrix}\mathbf{R'R}(\mathbf{R'ZQ}) = \begin{pmatrix}
\mathbf{Z'RR'Z} & \mathbf{Z'RQ}\\
\mathbf{Q'R'Z} & \mathbf{Q'Q}
\end{pmatrix} = \begin{pmatrix}
\mathbf{Z'RR'Z} &  \mathbf{0} \\
\mathbf{0} & \mathbf{I}
\end{pmatrix}
\end{equation*}
given that $\mathbf{Q} \bot \mathbf{R'Z}$. We can conclude that $\mathbf{V}_1$ and $\mathbf{V}_2$ are uncorrelated and independent due to normality, consequently $\mathbf{W}_1$ and $\mathbf{W}_2$ are independent as desired.
\end{proof}

When $\mathbf{A} \sim \mathcal{W}_n(\mathbf{\Sigma}, m)$ and $\mathbf{B} \sim \mathcal{W}_n(\mathbf{\Sigma}, d)$ are independent, $\mathbf{\Sigma}$ is an $n \times n$ positive definite matrix, $m \geq n$, the Wilks' statistic is
\begin{equation}
\Lambda = \frac{|\mathbf{A}|}{|\mathbf{A} + \mathbf{B}|} = |\mathbf{I}_n + \mathbf{A}^{-1} + \mathbf{B}|^{-1} \sim \mathcal{W}(n, m, d)
\end{equation}
and has a Wilks' $\Lambda$-distribution with parameters $n$, $m$ and $d$, see \cite{Mardia1979} and \cite{anderson2003} for a discussion on the Wilks' $\Lambda$ distribution. Setting $\mathbf{A} = \mathbf{W}_2$ and $\mathbf{B} = \mathbf{W}_2$, the test statistic can be written as
\begin{equation}\label{eq:Wilks}
\Lambda = \frac{|\mathbf{W}_2|}{|\mathbf{W}_2 + \mathbf{W}_1|} = \frac{|\mathbf{RSS}_1|}{|\mathbf{RSS}_0|}
\end{equation}
which, under linearity, follows a Wilks' $\Lambda$-distribution $\mathcal{L}(n, T - \text{cd}(\mathbf{X}) - \text{cd}(\mathbf{Z}), \text{cd}(\mathbf{Z}))$.  If $T$ is large, the Bartlett's approximation can be used
\begin{equation*}
\lambda = \left(\frac{1}{2}(n + \text{cd}(\mathbf{Z}) + 1)  + \text{cd}(\mathbf{X}) - T\right) \log \Lambda \sim \chi^2_{\text{cd}(\mathbf{Z})n},
\end{equation*}
see \cite{Bartlett1954} and \cite{anderson2003}.

The test statistic can be carried out after performing steps 1 and 2 in the algorithm in Section \ref{sec:Linearity}. Since matrix $\mathbf{Z}$ can be set to any $\mathbf{Z}$ in the sequential test procedure in Section \ref{sec:Nregimes}, the improved test statistic can be computed for both the tests used in this paper for the sequential procedure.

\end{document}